\documentclass[12pt,letterpaper]{article}
\usepackage{color}
\usepackage{graphicx}
\usepackage{epsfig,subfigure}
\usepackage{amssymb}
\usepackage{amsthm}
\usepackage{amsmath}
\usepackage[margin=1in]{geometry}
\usepackage{float}

\usepackage{algorithm}
\usepackage{algorithmic}

\usepackage{hyperref}

\usepackage{bbm,epsfig}
\usepackage{dsfont}
\usepackage{verbatim}
\usepackage{url}

\usepackage{algorithm, algorithmic, amsmath, amssymb, graphicx}
\usepackage{epstopdf}
\usepackage{enumerate}

\title{\LARGE \bf Distinguishing Infections on Different Graph Topologies
}

\author{Chris Milling, Constantine Caramanis, Shie Mannor and Sanjay
  Shakkottai\thanks{C. Milling, C. Caramanis and S. Shakkottai are
    with the Department of Electrical and Computer Engineering, The
    University of Texas at Austin, USA, Emails: {\tt
      cmilling@utexas.edu, constantine@utexas.edu,
      shakkott@austin.utexas.edu}. S. Mannor is with the Department of
    Electrical Engineering, Technion, Israel, Email: {\tt
      shie@ee.technion.ac.il}. This work was partially supported by
    NSF Grants CNS-1017525, CNS-0721380, EFRI-0735905, EECS-1056028,
    DTRA grant HDTRA 1-08-0029 and Army Research Office Grant
    W911NF-11-1-0265. Early versions of this paper have appeared in the Proceedings of ACM Sigmetrics, June 2012 \cite{sigmetrics2012}, and the Proceedings of the 50th Annual Allerton Conference on Communication, Control, and Computing, October 2012 \cite{allerton2012}.}}

\begin{document}

\providecommand{\card}[1]{\left|#1\right|}
\providecommand{\norm}[1]{\left\|#1\right\|}
\providecommand{\floor}[1]{\left\lfloor#1\right\rfloor}
\providecommand{\vect}[1]{\textrm{vec}(#1)}

\newcommand{\bydef}{\stackrel{\triangle}{=}}

\newcommand{\todo}[1]{\vspace{5 mm}\par \noindent \textsc{ToDo}
\framebox{\begin{minipage}[c]{0.35 \textwidth}
\tt #1 \end{minipage}}\vspace{5 mm}\par}

\newtheorem{theorem}{Theorem}
\newtheorem{proposition}{Proposition}
\newtheorem{lemma}{Lemma}
\newtheorem{corollary}{Corollary}
\newtheorem{definition}{Definition}

\newtheorem{thmsect}{Theorem}[section]


\maketitle

\begin{abstract}
The history of infections and epidemics holds famous examples where understanding, containing and ultimately treating an outbreak began with understanding its mode of spread. Influenza, HIV and most computer viruses, spread person to person, device to device, through contact networks; Cholera, Cancer, and seasonal allergies, on the other hand, do not. In this paper we study two fundamental questions of detection: first, given a snapshot view of a (perhaps vanishingly small) fraction of those infected, under what conditions is an epidemic spreading via contact (e.g., Influenza), distinguishable from a ``random illness'' operating independently of any contact network (e.g., seasonal allergies); second, if we do have an epidemic, under what conditions is it possible to determine which network of interactions is the main cause of the spread -- the {\em causative network} -- without any knowledge of the epidemic, other than the identity of a minuscule subsample of infected nodes? 

The core, therefore, of this paper, is to obtain an understanding of the {\em diagnostic power of network information}. We derive sufficient conditions networks must satisfy for these problems to be identifiable, and produce efficient, highly scalable algorithms that solve these problems. We show that the identifiability condition we give is fairly mild, and in particular, is satisfied by two common graph topologies: the grid, and the Erd\"os-Renyi graphs.\end{abstract}

\section{Introduction}
\label{sec:intro}

People and devices routinely interact through multiple networks -- contact networks -- be they virtual, technological or physical, allowing the rapid exchange of ideas, fashions, rumors, but also viruses and disease. Throughout this paper we refer to anything that spreads over a contact network as an {\em epidemic}. Understanding if something is indeed an epidemic best described through contact-network spreading, and secondly, understanding the {\em causative network} of that epidemic, is of critical important in many domains. Economists, sociologists and marketing departments alike have long sought to understand how ideas, memes, fads and fashions, spread through social networks. Meanwhile, epidemiology has understood the value of knowing the causative network of disease epidemics, from Influenza to HIV. Indeed, at one point, HIV was known as the ``4H disease'' where 4H referred to ``Haitians, Homosexuals, Hemophiliacs, and Heroin users'' \cite{aids-wiki,science06}. Understanding the causative network has greatly contributed to controlling the worldwide spread of the virus. 

While smartphone viruses have not yet supplanted computer viruses as the spreading technological threat of the hour, their potential for broad destructive impact is clear. Just as different human viruses may have different dominant spreading networks (again, compare Influenza and HIV), so may smartphone viruses spread over multiple networks, including bluetooth, SMS/MMS messaging, or e-mail.

A first step towards containing epidemics, be they technological or physical, relies on properly understanding the phenomenon as an epidemic in the first place, and then, accurately understanding the causative spread, before then adopting network-specific strategies for containment, quarantining and treatment.

Many factors complicate the process of determining the causative network. First, possibly because of long latency/hybernation periods, variation in reporting/detection, or simply lack of data, in some cases it may be difficult or impossible to collect accurate longitudinal data. Equally importantly, the reporting set of those ``infected'' (be they people or devices) may be only a tiny fraction of those in fact infected. Therefore in this paper, we consider the most dire information regime: we assume we have data from only a single snapshot of time, where only a (perhaps vanishing) fraction of the infected population reports.

With these data, this paper focuses on determining the causative network for the spread of an epidemic (e.g. virus, sickness, or opinion) from limited samples of the network state.

\subsection{Setting and Results}
\label{sec:setting}

We model people/devices/etc.\ as a set of nodes, $V$, of a graph. The nodes in $V$ become infected by an epidemic that spreads according to either graph $G_1 = (V,E_1)$, or $G_2 = (V,E_2)$, propagating along the edges of these graphs, according to an SI model of infection \cite{massganesh05:epidemics}. Given a (potentially small) sub-sample of the infected nodes at a single snapshot in time, our objective is to determine the network over which the epidemic is spreading. If one of the graphs, say $G_2$, is a star graph, where each node has a single edge to an external infection source, this models the problem of distinguishing an epidemic spreading on $G_1$, from a random illness spreading according to no network structure.

This paper is about understanding when the two processes -- spreading on $G_1$ or $G_2$ -- are statistically distinguishable, and moreover when this can be done {\em by an efficient algorithm}. Evidently, in certain regimes, no algorithm can distinguish between the two processes. First, the graphs need to be sufficiently different. We quantify this precisely in Section~\ref{sec:model-algo}. Beyond this, certainly, if (almost) everyone is infected, or if (almost) none of those infected report, then nothing can be done. Our results are presented in terms of these two quantities: we are interested in understanding the maximum number of nodes (people/devices) that can be infected, and simultaneously the minimum number of these that actually report they are infected, so that our algorithms correctly distinguish the true spreading process, with high probability. 

There are two regimes of graph topologies we consider: the setting where $G_2$ is a star graph -- we call this the `infection vs. random sickness' problem -- and then the setting where both $G_1$ and $G_2$ exhibit nontrivial network structure -- we call this the `graph comparison' problem. For the sake of the mathematical exposition, we find it more natural to present first the graph comparison problem, and then the infection vs. random sickness problem.

We provide efficiently computable algorithms to answer the above questions, and then provide sufficient conditions on the regimes where our algorithms are guaranteed to succeed, with high probability. Specifically, our main contributions are as follows:

\begin{itemize}

\item[(i)] \textbf{Algorithm:} We develop efficiently computable algorithms for both problems. For inferring the causative network in the graph comparison problem, we develop what we call the Comparative Ball Algorithm. For the `infection vs. random sickness', we develop two algorithms: the Threshold Ball Algorithm and the Threshold Tree Algorithm.  These algorithms build on the intuition that infected nodes are clustered more strongly on the true causative network. If on one network, the clustering is tighter, it is more likely that it is driving the infection. We quantify clustering based on the ball radius that contains the infected nodes.

\item[(ii)] \textbf{Guarantees for General Graphs:} For the graph comparison problem, we identify two natural graph conditions that we use to give very general performance guarantees for our Comparative Ball Algorithm. The first property is called the {\em (a) Speed condition}; a graph satisfies this if the epidemic ball radius increases linearly in time. The second key property is called the {\em (b) Spread condition}; a graph satisfies this if a randomly selected collection of nodes are sufficiently spread apart, with respect to the natural metric induced by the graph. For any two graphs that satisfy both {\em (a)} and {\em (b)}, we derive upper bounds on the number of total infected nodes, and lower bounds on the number of reporting nodes, so that our Comparative Ball Algorithm is guaranteed to correctly determine the causative network (as $n \to \infty$ and with high probability).

\item[(iii)] \textbf{Grids and the Erd\"{o}s-Renyi Random Graphs:} For both $d$-dimensional grids, and the giant component of the Erd\"{o}s-Renyi random graph (with constant asymptotic average degree), and for both the graph comparison and infection vs. random sickness problem, we derive bounds on the parameters associated with the speed and spread conditions, thus, providing sufficient conditions on the regime where we can determine the causative network.

\end{itemize}

\subsection{Related Work}

The infection model we consider in this paper is the susceptible-infected (SI) model where nodes transition from {\em susceptible} to {\em infected} according to a memoryless process \cite{massganesh05:epidemics}. Much of the work on this model has focused on the predictive or analytic side, focused on characterizing the spread of the infection under various different settings. For example, \cite{ball04} considers graphs with multiple mixing distances (that is, local and global spreading), while \cite{gopalan11} considers the setting where the infected nodes are mobile. There are other approaches to modeling infection, and while interesting to extend the current ideas and analysis there, we do not consider these in the present work. 


Our work, in contrast, lies on the inference side, where given (partial) information about the realization of an epidemic, the goal is to infer various properties or parameters of the spreading process. While quite different in terms motivation and goals, a few recent works have also considered epidemic inference. In \cite{demiris05a}, the authors provide a Bayesian inference approach for estimating the transmission rates of the infection. Alternatively, one can use MCMC methods to estimate the model parameters \cite{streftaris02}, \cite{demiris05b}. A similar problem is considered in \cite{infectionsource,shza11}, where, given a set of infected nodes, one seeks to determine which node is most likely to be the original source of the infection. 

On the technical side, several of our results are related to first-passage percolation. In the first-passage percolation basic formulation, there is a (lattice) graph of infinite size. For each edge, an independent random variable is generated that represents the time taken to traverse that edge. Some node is denoted as the source, and the time taken to reach another node is the minimum of the total time to traverse a path over all paths between the source and that destination. This is equivalent to an infection traveling through the network as considered here. Work has been done to analyze various characterizing properties of this percolation, such as the {\em shape} of the infection and the {\em rate} at which it spreads. In the sequel, we find particularly useful percolation results on trees \cite{rwreandfpptrees} and lattices \cite{gridinfectionlimit}.

\subsection{Outline of the Paper}
The paper is organized as follows. In Section \ref{sec:model-algo}, we define precisely the infection model as well our two main problems: determining the causative infection network between two graphs, and between a graph and a random sickness. Section \ref{sec:graph_compare} contains our analysis of the problem of distinguishing infections between two different graphs. We provide an efficient algorithm, and then the success criteria of this algorithm for distinguishing between epidemics on general graphs. We show that the sufficient conditions we provide are satisfied by a general class of graphs, that include two standard graph topologies, $d$-dimensional grids and Erd\"{o}s-Renyi graphs. Then, in Section \ref{sec:infectvsrandom}, we turn to the problem of distinguishing an infection from a random sickness. Recall that this is equivalent to taking one of the two graphs to be the star graph. Star graphs, however, do not have non-trivial neighborhoods, and hence the algorithm and analysis from the previous part do not immediately carry over. We develop two new algorithms for this setting, and provide success guarantees for each. We consider grid and Erd\"{o}s-Renyi graphs. Finally, Section \ref{sec:simulations} contains the simulations data for each of these problems and illustrates the empirical performance of our algorithm on these graphs. Our results demonstrate that on synthetic data, empirical performance recovers the theoretical results. We also test our algorithms on a real-world graph, and our simulations show that here too, our algorithms are quite effective.


\section{The Model}
\label{sec:model-algo}

We consider a collection of $n$ nodes (vertices $V$) which are members
of two different networks (graphs). These graphs are denoted by $G_1 =
(V, E_1)$ and $G_2 = (V, E_2)$; they share the same vertex set but have
different edge sets. For example, $G_1$ could represent the $n$
vertices arranged on a $d-$dimensional grid, and $G_2$ could be an
Erd\"{o}s-Renyi graph. Note that $G_2$ does {\em not} need to have
qualitatively different structure from $G_1$: Indeed $G_2$ could also
be a $d-$dimensional grid, but with a different node-to-edge mapping.

\subsection{Objective}
We assume that the two graph topologies, $G_1$ and $G_2$ are known. At some point in time, an epidemic begins at a random node and spreads according to the edges of one of the two graphs, following the infection model described below in Section \ref{ssec:infectionmodel}. At some snapshot in time, a small random subset of the infected nodes report their infection. From the knowledge of the graph topologies and the identity of the reporting nodes (but without knowledge of the other infected nodes) our objective is to design an algorithm that (asymptotically, as the size of the problem scales) correctly determines which graph the epidemic is spreading on.

We first study the setting where both $G_1$ and $G_2$ have non-trivial neighborhoods, and the goal is to detect which graph is responsible for spreading the epidemic; we call this the Graph Comparison Problem. We then consider the setting where $G_2$ is the star graph, hence modeling the problem of distinguishing an epidemic from a random illness.


\subsection{Infection Model}
\label{ssec:infectionmodel}
We assume that an epidemic propagating on one of the two graphs, $G_1$ or $G_2$. The objective is to determine on which network it is spreading. We reiterate that this `epidemic' could model many situations, including the spread of a cellphone virus, physical sickness of humans, and opinions or influence about products or ideas.

Given that the epidemic is on graph $G_i,$ the spread occurs as
follows (the standard SI dynamics \cite{massganesh05:epidemics}). A
node is randomly selected to be the epidemic seed, and thus is the first ``infected'' node. At random times, the illness spreads from the sick nodes to some subset of the neighbors of the sick nodes, according to an exponential process. Specifically, associate an independent mean 1 exponential random variable with each edge incident to an infected and an uninfected (a susceptible) node. The realization of this random variable represents the transit time of the infection across that specific edge -- a random variable. Thus an infected node proceeds to infect its neighbors, with each non-infected neighbor becoming infected after the random transit time associated with the edge between the infected node and this neighbor. This process proceeds until the entire graph $G_i$ is infected.

If the graph is a star graph, then every node is incident to a single external node. Consequently, nodes become sick at the same rate, and independently of every other node. This process, then, is stochastically equivalent to a random illness, where by a given time $t$, each node has become sick independently with some fixed probability $\hat{q}$.

In either case, the infection continues until some (unknown) time $t$. At this time, a sub-sample of the infected nodes report their infection state independently, each with some probability $q < 1.$ We let $S$ denote the set of infected nodes, and $S_{{\rm rep}} \subseteq S$ the set of reporting infected nodes.\footnote{Note that we suppress the dependence of both $S$ and $S_{{\rm rep}}$ on $n$, unless required for clarity.}

\subsection{Graph Structure}
For the statistical problem of distinguishing the causative network to be well-posed, the contact networks encoded by graphs $G_1$ and $G_2$ must be sufficiently different. Note that this does not imply that the topology of the graphs must be different (indeed, it could be identical). Rather, the neighborhoods of each graph must be distinct, i.e., the nodes that are near an infected node with respect to one graph, must be different from the nodes near the same infected node, with respect to the other graph. We note that if this is not the case, then both graphs encode approximately the same causative network, and hence solving the comparative graph problem is not that important.

In this paper, we require that corresponding nodes on the two graphs have {\em independent neighborhoods}.\footnote{We note that we can envision other conditions based on clustering of epidemics on the two graphs which also serves as alternate sufficient conditions. For simplicity, we restrict ourselves to the `random node index' condition in this paper.}  It is easiest to explain this condition by means of a random construction, which is also the one we assume for the results in the sequel. Let $G_1$ and $G_2$ be graphs of the same size $n$, {\em whose nodes are unlabeled}. Then randomly label the nodes of graph $G_1$ from `1' to `$n$' uniformly, and independently and uniformly at random label the nodes of graph $G_2$. Nodes of the same label represent the same entity (person, device), i.e., if a node on one graph is infected, the corresponding node on the other graph is also infected.

This independent neighborhood condition approximately holds in typical settings. Consider for instance the several hundred ``nodes'' (people, or devices) that come within blue-tooth range during a walk through the mall. This list likely has extremely small overlap (possibly only the few friends accompanying us on the mall excursion) with the set of nodes that send us e-mail or SMS on a regular basis.

\section{Graph Comparison Problem}
\label{sec:graph_compare}
The graph comparison problem consists of distinguishing the causative graph for an infection spreading on one of two {\em structured} graphs $G_1$ and $G_2$. We make precise what we mean by {\em structured graphs} below, but intuitively, both graphs have non-trivial neighborhood structure, in contrast to the star graph. This is the key technical feature that differentiates the comparative graph problem from the infection vs. random sickness problem, which we take up in Section \ref{sec:infectvsrandom}. As the algorithm reveals, the key in the comparative graph problem is that, under appropriate conditions, the infection, or epidemic, is clustered on either $G_1$ or $G_2$. In the case where $G_2$ is the star graph, there is no notion of clustering there, so our algorithms must detect clustering vs. absence of clustering.

We turn to the details of the comparative graph problem. The first order of business is understanding precisely what conditions we require the topology of graphs $G_1$ and $G_2$ to satisfy, making precise the notion of ``non-trivial neighborhood structure'' that, where, unlike the star graph, an epidemic exhibits some statistically detectable clustering. There are two key properties required: first, the infection must spread at a bounded speed; second, a random collection of nodes on the graph must, with high probability, not exhibit a strong clustering. Of course, the star graph fails with respect to the minimum spread of random nodes condition. As another example that fails the bounded speed condition, consider a tree whose nodes have degree $d^{k+1}$ at level $k$. 

We now state these conditions precisely, and in addition, we show, many graphs satisfy these conditions, including familiar topologies like the $d$-dimensional grid and the Erd\"os-Renyi graphs. It is also easy to see that any graph with bounded degree also satisfies these two conditions. 

We need first a simple definition:
\begin{definition}
Given a graph $G=(V,E)$ and a subset of its nodes, $S \subseteq V$, let ${\rm RadiusBall}(G,S)$ denote the radius of the smallest ball that contains $S$. This can be computed in time at most $O({\rm card}(V)^2)$. 
\end{definition}

Let $\mathcal{G} = \{\mathcal{G}^{(n)}\}$ denote a family of graphs, where $\mathcal{G}^{(n)}$ denotes the subset of the graphs of $\mathcal{G}$ that have $n$ nodes. For each $n$, there is a (possibly trivial) probability space $\left(\mathcal{G}^{(n)}, \sigma(\mathcal{G}^{(n)}), P^{(n)} \right)$. Concrete examples include the set of $d$-dimensional grid graphs, Erd\"os-Renyi graphs with bounded expected degree, $d$-regular trees, etc. 

\begin{definition} A family $\mathcal{G}$ satisfies the {\em speed} and {\em spread} conditions, if there exist constants $s_{\mathcal{G}}$, $b_{\mathcal{G}}$ and $\beta_{\mathcal{G}}$, such that for any sequence $\{G^{(n)}\}$ picked randomly from the product probability space $\prod_n \mathcal{G}^{(n)}$, 
the following hold with probability approaching $1$ as $n$ increases, where the probability is over the random subset of nodes in the definitions below, and, in the case of random families, $\mathcal{G}$, such as Erd\"os-Renyi graphs, over the selection of $G^{(n)}$ as well: 

\begin{itemize}
\item[] \textbf{Speed Condition}: For infections starting at a randomly selected nodes and infection times $t^{(n)} \rightarrow \infty$, the set $S^{(n)}$ of nodes infected at time $t^{(n)}$ satisfies ${\rm RadiusBall}(G^{(n)}, S^{(n)}) < s_{\mathcal{G}} t^{(n)}.$ 

\item[] \textbf{Spread Condition}: First, ${\rm diam}(G^{(n)} = \Omega(\log n)$. Second, a random set $S^{(n)}$ of nodes of $G^{(n)}$, with
  ${\rm card}(S^{(n)}) > \beta_{\mathcal{G}} \log n$, satisfies ${\rm RadiusBall}(G^{(n)}, S^{(n)}) > b_{\mathcal{G}} {\rm diam}(G^{(n)}).$

\end{itemize}
\end{definition}

These two conditions essentially encode the properties required so that an infection spreading on a graph $G_1^{(n)}$ (chosen from family $\mathcal{G}_1$) exhibits clustering, and, conversely, if it is spreading on another graph $G_2^{(n)}$ (chosen from family $\mathcal{G}_2$) with independent neighborhoods (as described above) then there is no clustering with respect to $G_1^{(n)}$.

Note that to ease notation, whenever the context is clear, we drop the superscript $(n)$ that denotes the number of nodes.
 
\subsection{The Comparative Ball Algorithm}
\label{ssec:compballalgo}
We provide an algorithm for the Comparative Graph Problem, called the {\em Comparative Ball Agorithm}, and then give a theorem with sufficient conditions guaranteeing its success. The algorithm is natural, given the discussion above. We find the smallest ball on that graph that contains all the reporting infected nodes. We take the ratio of the radius of this ball to that of the graph's diameter. These ratios -- called the {\em score} of each graph -- serve as a topology independent measure of clustering on each graph. The Comparative Ball Algorithm returns the graph with the smallest normalized clustering ratio. This is formally described below.

To specify our algorithm precisely, we require the following definitions. Given a graph $G,$ a node $v$, and a radius $r$, we denote by $Ball_{v,r}(G)$ the collection of all nodes on the graph $G$ that are at most a distance $r$ from node $v$ (graph distance measured by hop-count). As we have done above, we denote the diameter of the graph by ${\rm diam}(G).$ Given any collection of nodes $S,$ we \ denote by $Ball(G, S)$ the smallest-radius ball that contains all the nodes in $S,$ and we use ${\rm RadiusBall}(G, S)$ as in the definition above, to denote its corresponding radius.

\begin{algorithm}
\caption{Comparative Ball Algorithm}
\textbf{Input:} Two graphs, $G_1$ and $G_2$; Set of reporting infected
nodes $S_{{\rm rep}}$;\\
\textbf{Output:} $G_1$ or $G_2$ \\
\label{alg:comp_ball}
\begin{algorithmic}
\STATE $a_1 \gets {\rm RadiusBall}(G_1,S_{{\rm rep}})$
\STATE $b_1 \gets {\rm diam}(G_1)$
\STATE $x_1 \gets a_1/b_1$
\STATE $a_2 \gets {\rm RadiusBall}(G_2,S_{{\rm rep}})$
\STATE $b_2 \gets {\rm diam}(G_2)$
\STATE $x_2 \gets a_2/b_2$
\IF{$x_1 \leq x_2$}
\RETURN {$G_1$}
\ELSE
\RETURN {$G_2$}
\ENDIF
\end{algorithmic}
\end{algorithm}

\subsection{Main Result: General Graphs}
\label{ssec:general}

We prove that if $\mathcal{G}_1$ and $\mathcal{G}_2$ satisfy the speed and spread conditions given above (i.e., they have finite speed and spread constants), then the Comparative Ball Algorithm can distinguish infections on any two such graphs (with probability 1, as $n \rightarrow \infty$). The speed and spread conditions turn out to be fairly mild. In Section \ref{ssec:topologies} we show that, among many others, two commonly encountered, standard types of graphs satisfy these properties: $d-$dimensional grids and Erd\"{o}s-Renyi graphs. The proof that Erd\"{o}s-Renyi graphs satisfy the speed and spread conditions immediately implies that bounded-degree graphs also satisfy speed and spread conditions.

Our results are probabilistic, guaranteeing correct detection with probability approaching 1, as the number of nodes $n$ in the graphs (recall the vertex sets of the two graphs are the same -- it is on these nodes that the infection is spreading) scales. Therefore, our results are properly stated on a pair of families of graphs, $\{(G_1^{(n)},G_2^{(n)})\}$, where each $G_1^{(n)}$ comes from some family $\mathcal{G}_1$, and similarly for $\mathcal{G}_2$. For notational simplicity, we refer simply to $G_1$ and $G_2$ to denote both specific graphs in this sequence, and the entire sequence as well. Thus, by ${\rm diam}(G_1)$ we mean the diameter of the specific graph $G_1^{(n)}$, hence this is a value that depends on $n$, where as the quantities $s_{\mathcal{G}_1}$, $b_{\mathcal{G}_1}$ and $\beta_{\mathcal{G}_1}$ depend on the family, and are independent of $n$. The infection time is $t^{(n)}$, and we require $t^{(n)} \rightarrow \infty$. Like for the graphs, we drop the superscript for clarity and use $t$ to denote the infection time.
 
\begin{thmsect}
Consider families of graphs $\mathcal{G}_1$ and $\mathcal{G}_2$ satisfying the speed and spread conditions above, and let $\{(G_1^{(n)},G_2^{(n)})\}$ denote a sequence of graphs drawn from $\mathcal{G}_1$ and $\mathcal{G}_2$. Consider infection times $t^{(n)}$ such that the number of reporting infected nodes scales at least as $\max(\beta_{\mathcal{G}_1}, \beta_{\mathcal{G}_2}) \log n$. Then if $t < b_{\mathcal{G}_2} {\rm diam}(G_1)/s_{\mathcal{G}_1}$, the Comparative Ball Algorithm correctly identifies an infection on $G_1$ with probability approaching $1$. In addition, if $t < b_{\mathcal{G}_1} {\rm diam}(G_2)/s_{\mathcal{G}_2}$, then the Comparative Ball Algorithm correctly identifies an infection on $G_2$ with probability approaching $1$. 
\end{thmsect}
\begin{proof}
By symmetry, it is sufficient to prove that an infection is detected on $G_1$. For every $n$, let $S_{{\rm rep}}$ (again we suppress dependence on $n$ when it is clear from the context) denote the set of reporting sick nodes, where ${\rm card}(S_{{\rm rep}}) > \beta_{\mathcal{G}_2} \log n$. Note that by the independence assumption, this set of nodes is randomly distributed over $G_2$. By the speed and spread conditions, with probability approaching $1$ as $n$ scales, ${\rm RadiusBall}(G_1, S_{{\rm rep}}) < s_{\mathcal{G}_1} t$ and
${\rm RadiusBall}(G_2, S) > b_{\mathcal{G}_2} {\rm diam}(G_2)$. Then the score for the first
graph satisfies $x_1 < s_{\mathcal{G}_1} t/{\rm diam}(G_1) < b_{\mathcal{G}_2}$ by
hypothesis. Similarly, $x_2 > b_{\mathcal{G}_2} {\rm diam}(G_2)/{\rm diam}(G_2) =
b_{\mathcal{G}_2}$. Therefore, the algorithm correctly identifies an infection.
\end{proof}

\subsection{Speed and Spread Conditions: Grids and the Erd\"{o}s-Renyi
  Graph}
\label{ssec:topologies}

In this section we show that the spread and speed conditions are fairly mild, by demonstrating that they hold on two common types of graphs: the $d$-dimensional grid, and the Erd\"{o}s-Renyi graph.  The $d$-dimensional grid graph is an example of a contact graph where the infection spreads between nodes in spatial proximity (e.g., the Bluetooth virus, human sickness). The second topology is an Erd\"{o}s-Renyi graph, a random graph forming a network with low diameter. This topology models an infection spreading over long distance networks, such as the Internet or over social networks. We show that both of these networks satisfy the spread and speed conditions, and hence that the Comparative Ball Algorithm successfully determines the causative network on these graphs. As mentioned above, our proofs for the Erd\"{o}s-Renyi graphs immediately carry over to all bounded-degree graphs.

\subsubsection{$d-$Dimensional Grids}
\label{sssec:grid}

Let the graph $G = {\rm Grid}(n, d)$ be a grid network with $n$ nodes and dimension $d$, so the side length is $n^{1/d}$. We avoid edge effects by wrapping around the grid (a torus). This avoids dealing with non-essential complexities resulting from the choice of the initial source of the infection.

First, we establish limits on the speed of the infection after time $t$ has passed. Next, we show lower bounds on the spread, i.e., the ball size needed to cover a random selection of nodes of sufficient size. Together, these show that grid graphs satisfy the speed and spread conditions.

Since we model the time it takes the infection to traverse an edge as an independent exponentially distributed random variable, the time a node is infected is the minimum sum of these random variables over all paths between the infection origin and that node.  This simply phrases the infection process in terms of first-passage percolation on this graph. This allows us to use a result characterizing the `shape' of an infection on this graph (see \cite{gridinfectionlimit}). Let $I(t)$ be the set of infected nodes at time $t$. Identifying the nodes of the graph with points on the integer lattice embedded in $\mathbb{R}^d$ with the infection starting at the origin, let us put a small $\ell^{\infty}$-ball around each infected node. This allows us to simply state inner and outer bounds for the shape of the infection. To this end, define this expanded set as $B(t) = I(t)+ [-1/2,1/2]^d$. 

\begin{lemma}[\cite{gridinfectionlimit}]
\label{lem:gridbound}
There exists a set $B_0$ and constants $C_1$ to $C_5$ such that for $x \leq \sqrt{t}$,
$$
P \{ B(t)/t \subset (1+x/\sqrt{t}) B_0 \} \geq 1-C_1t^{2d}e^{-C_2 x}
$$
and
\begin{align}
P & \{ (1-C_3t^{-1/(2d+4)}(\log{t})^{1/(d+2)}) B_0 \subset B(t)/t \}\nonumber\\
& \geq 1-C_4t^d \exp{(-C_5t^{(d+1)/(2d+4)}(\log{t})^{1/(d+2)})}. \nonumber
\end{align}
\end{lemma}

That is, the shape of the infected set $B(t)$ can be well-approximated by the region $t B_0$.

Moreover, one can show that this set $B_0$ is regular in that it contains an $\ell^1$-ball and is contained in an $\ell^{\infty}$ ball: $\{ x : \norm{x}_1 \leq \mu \} \subset B_0 \subset [-\mu, \mu]^d$, where $\mu \bydef \sup_x \{(x, 0, ..., 0) \in B_0\}$, effectively the rate the infection spreads along an axis \cite{gridinfectionlimit}. Note that $\mu$ does not depend on the {\em realization} of the process, only the statistics of the spread. We use this result to establish the outer bound of the infection.

\begin{proposition}
\label{prop:gridInfectUB}
Let $G^{(n)} = {\rm Grid}(n, d)$ and let $t^{(n)}$ denote any sequence of increasing times, $t^{(n)} \rightarrow \infty$. As defined above, $S_{{\rm rep}}^{(n)}$, denotes the (random) subset of nodes infected by the epidemic, that report their infected status. Then there exists a constant $\mu$ such that 
$$
{\rm RadiusBall}(G^{(n)},S_{{\rm rep}}^{(n)}) < 1.1 d \mu t^{(n)},
$$
with probability converging to $1$ as $n \rightarrow \infty$.
\end{proposition}

\begin{proof}
We drop the indexing w.r.t. $n$, since the context is clear. Let $\mu \bydef \sup_x \{(x, 0, ..., 0) \in B_0\}$ and $m = 1.1 d \mu t$. Then we must show ${\rm RadiusBall}(G,S_{{\rm rep}}) < m$ with probability approaching $1$. Note that if the infection can be limited to the subgrid $[-m/d, m/d]^d$ (with appropriate translations), then this condition is satisfied. Define $E$ as the event that ${\rm RadiusBall}(G,S_{{\rm rep}}) \geq m$. Therefore, using Lemma \ref{lem:gridbound},

\begin{align}
P(E) &< 1-P \{ B(t) \subset [-m/d, m/d]^d \}\nonumber\\
&< C_1t^{2d}e^{-C_2 t^{-1/2} (m/(d\mu) - t)} \nonumber\\
&= C_1t^{2d}e^{-0.1 C_2 t^{1/2}} \nonumber\\
&\rightarrow 0. \nonumber
\end{align}
Hence, we see that ${\rm RadiusBall}(G,S_{{\rm rep}})$ satisfies the required bound with high probability.
\end{proof}

The following theorem provides a lower bound on the radius of the ball
needed to cover a collection of random nodes uniformly selected from
the grid. We require that the number of random nodes grows at least as $\log
n$.

\begin{proposition}
\label{prop:gridRandLB}
Let $G^{(n)} = {\rm Grid}(n, d)$. Let $S^{(n)}$ be a collection of nodes chosen uniformly at random from
$G^{(n)}$, such that ${\rm card}(S^{(n)}) > \log n$ for sufficiently high $n$. Then 
$$
{\rm RadiusBall}(G^{(n)},S^{(n)}) > n^{1/d}/4,
$$
with probability converging to $1$ as $n \rightarrow \infty$.
\end{proposition}

\begin{proof}
Again we drop the $n$-index wherever context makes it clear. By assumption, we have a set $S$ of random nodes with
${\rm card}(S) > \log n$. Define $X = {\rm card}(S)$. We show the probability all nodes in $S$ are within some ball of radius $n^{1/d}/4$ decays to $0$ with
$n$. Then consider one of the $n$ such balls. There are less than $l =
(n^{1/d}/2)^d$ nodes in that region (the number of nodes in a `box' of
side $n^{1/d}/2$). Within this ball, there are at most
$\binom{l}{X}$ arrangements of the sick nodes out of $\binom{n}{X}$
total possible arrangements. Therefore, the probability all the sick
nodes are within the region is no more than 
\begin{align}
\binom{l}{X} \Big/ \binom{n}{X} & = \frac{l!(n-X)!}{(l-X)!n!}\nonumber\\
 & \leq (l/n)^{X}.\nonumber
\end{align}

Using a union bound, we find that the probability there is a ball of
that size containing all nodes in $S$ is at most $n (l/n)^{X}$. Then 
\begin{align}
n (l/n)^{X} & < n \left( \frac{1}{2^d} \right)^{\log n} \nonumber\\
 & = n^{1 - d \log 2} \nonumber\\
 & \rightarrow 0.\nonumber
\end{align}
Therefore, ${\rm RadiusBall}(G, S) > n^{1/d}/4$ with probability converging to $1$.
\end{proof}

Since the diameter of a grid is (nearly) $d/2 n^{1/d}$, we see
that a grid satisfies both the speed condition (Proposition
\ref{prop:gridInfectUB}) and the spread condition (Proposition
\ref{prop:gridRandLB}), and hence the
Comparative Ball Algorithm performs well on grid graphs.

\subsubsection{Erd\"{o}s-Renyi Graphs and Bounded Degree Graphs}
\label{sssec:gnp}

Now we consider Erd\"{o}s-Renyi graphs, representing infections that spread over low diameter networks (the diameter grows logarithmically with network size). An Erd\"{o}s-Renyi graph is a random graph with $n$ nodes, where there is an edge between any pair of nodes, independently with probability $p.$ We study the Erd\"{o}s-Renyi graph in the regime where $p = c/n,$ for some positive constant $c > 1.$
This setting leads to a disconnected graph; however, there exists a
giant connected component with $\Theta(n)$ nodes with high probability
in the large $n$ regime. In this paper, we restrict our attention to
epidemics on this giant component. Thus we limit both the infection and the random set of reporting nodes (due
to the labeling when the infection occurs on the alternative graph) to
occur exclusively on the giant connected component. If the infection
on the other graph contains too many nodes for the giant component, we
simply ignore the excess, but this point is already outside the regime
of interest.

We establish two results in this section. We first prove an upper bound on the ball size for an infection up to a limited time, and next, we demonstrate a lower bound on the ball size for a random collection of nodes. 

Note that the two results given in this section also hold for bounded-degree graphs. The proofs immediately carry over to this class. For simplicity, and because the randomness of the Erd\"os-Renyi graphs presents some further complications, we state everything in terms of the Erd\"os-Renyi graphs.

\begin{proposition}
\label{prop:gnpInfectUB} 
Let $G^{(n)}$ denote the connected component of a realization of a $G(n, p)$ graph, and let the sequence $t^{(n)}$ denote increasing time instances, scaling (without bound) with $n$. As above, let $S_{{\rm rep}}^{(n)}$ denote the random subset of nodes reached by the epidemic, that also report. Then there exists a constant $C_6$ such that 
$$
{\rm RadiusBall}(G^{(n)},S_{{\rm rep}}) < C_6 t^{(n)},
$$
with probability converging to $1$ as $n \rightarrow \infty$.
\end{proposition}

\begin{proof}
Since the dependence on $n$ is clear, we drop the index of $n$. This theorem essentially states that there is a maximum speed at which the infection can travel on an Erd\"{o}s-Renyi graph. The
  statement follows from a similar maximum speed result for trees
  \cite{itai94}. Therefore, it remains to show how this result can be applied to an Erd\"{o}s-Renyi graph.
  To do this, we upper bound an infection on an Erd\"{o}s-Renyi
  graph by a tree that represents the routes on which an infection can
  travel. Since an Erd\"{o}s-Renyi graph is locally tree-like
  \cite{durrett07}, we expect this approximation to be fairly accurate
  for low times, though this is not necessary for the proof.
  
  Consider the tree $\tilde{G}$ formed as follows. The root of the tree is the
  initial infected node. The next level contains copies of all nodes
  adjacent to the original node in the Erd\"{o}s-Renyi graph. Each of
  these have descendants that are copies of their neighbors, and so on. Note all
  nodes may (and likely do) have multiple copies.

  We start an infection at the root of $\tilde{G}$ and let it spread for time $t$. Consider the induced set of infected nodes, $\tilde{S}_{{\rm rep}}$, as
  the set of nodes in $G$ which have copies that are infected on $\tilde{G}$. Since the
  distance of a copy from the root of $\tilde{G}$ is no less than the distance from
  the original node to the original infection source, we see that the distance the
  infection has traveled on $\tilde{G}$ is no less than the distance from the infection source to the
  farthest node in $\tilde{S}_{{\rm rep}}$ (on $G$). Note that the $\tilde{S}_{{\rm rep}}$ stochastically dominates the true infected set $S$. That is, for all sets $T$, $P(T \subset \tilde{S}_{{\rm rep}}) \geq P(T \subset S_{{\rm rep}})$. 

  This stochastic dominance result follows from the fact that the transition rates are universally
  equal or higher for the induced set. Hence, ${\rm RadiusBall}(G, S_{{\rm rep}})$ is also stochastically
  dominated by ${\rm RadiusBall}(G, \tilde{S}_{{\rm rep}})$, and the latter is upper
  bounded by the depth of the infection in the tree, which using the speed result, is
  bounded by $C_6 t$ for some speed $C_6$. That is, with probability tending to $1$,
  $$
  {\rm RadiusBall}(G,S_{{\rm rep}}) < C_6 t.
  $$
\end{proof}

Next, we use the neighborhood sizes on this graph to provide a lower
bound to the ball size needed to cover a random infection. 

\begin{proposition}
\label{prop:gnpRandLB}
Let $G^{(n)} = G(n, p)$, and let $S^{(n)}$ denote a collection nodes sampled uniformly at random from $G^{(n)}$,
such that ${\rm card}(S^{(n)})$ scales at least with $\log n$. Then 
$$
{\rm RadiusBall}(G^{(n)},S^{(n)}) > \frac{\log n}{3 \log c},
$$
with probability converging to $1$ as $n \rightarrow \infty$.
\end{proposition}

\begin{proof}
We suppress the index $n$ for clarity. We proceed by bounding the probability
that all the random nodes are within a ball of radius $m$. This is
possible only if all nodes in $S$ are within distance $2m$ from any
given node in $S$. Now, the number of nodes within a distance $2m$
from a given node is no more than $16 m^3 c^{2m} \log n$ with
probability $1-o(n^{-1})$ \cite{neighborhoodbound}. Then the
probability of all nodes fitting inside one such ball is at most  
$$
\left( \frac{16 m^3 c^{2m} \log n}{n} \right)^{{\rm card}(S)-1} < \left(
  \frac{16 m^3 c^{2m} \log n}{n} \right)^{\log n-1}. 
$$
Then this decays to $0$ at least as fast as $n^{-1}$ if
$$
\frac{16 m^3 c^{2m} \log n}{n} < n^{-1/\log n}.
$$
Finally we set $m = \frac{\log n}{3 \log c}$ as desired. Hence $c^{2m}
= n^{2/3}$. Using this substitution, the above term reduces to 
\begin{align}
\frac{16 m^3 c^{2m} \log n}{n} & = \frac{16 m^3 n^{2/3} \log n}{n} \nonumber\\
 & = \frac{16 (\log n)^4}{27 (\log c)^3 n^{1/3}} \nonumber\\
 & < (\log n)^4 n^{-1/3} < n^{-1/\log n}
\end{align}
for sufficiently large $n$. Therefore, ${\rm RadiusBall}(G,S) > \frac{\log
  n}{3 \log c}$ with probability converging to $1$. 
\end{proof}

The diameter of the giant component of an Erd\"{o}s-Renyi graph is
$\Theta(\log n/\log c)$ \cite{durrett07}. Thus, Propositions
\ref{prop:gnpInfectUB} and \ref{prop:gnpRandLB} establish that
an Erd\"{o}s-Renyi graph satisfies both the speed and spread
conditions respectively.

\section{Infection vs. Random Sickness}
\label{sec:infectvsrandom}

We now turn to the setting where $G_2$ is the star graph. This is the problem of distinguishing an epidemic spreading on a structured graph, from a random illness affecting any given node independently of the infection status of any of its neighbors. As discussed above, and as with the graph comparison problem, distinguishing these two modes of infection becomes difficult when many nodes are infected, and when only a small fraction of the infected nodes report their infection. 

For this problem, we label the structured graph $G$. In an infection, the sick nodes will be clustered on $G$. On the other hand, in the case of random illness, the infection is not guaranteed to exhibit clustering on any graph. Moreover, the star graph, of course, fails to satisfy the spread conditions. Therefore, the graph comparison algorithm and its analysis cannot suffice. Instead, we must find a test for the absence of clustering. It is most natural to use a simple threshold test for the degree of clustering. This threshold, however, itself depends on the parameters of the problem, in particular, the rate at which infected nodes report their condition (the parameter $q$), and the time elapsed since the epidemic began propagating, or, equivalently, the expected infection size. We consider first the setting where these parameters are explicitly known, and then turn to the setting where time (and hence, the expected infection size) is not known. In this case, we demonstrate that this can be estimated with sufficient accuracy, based on the reporting nodes.


\subsection{Threshold Algorithms}
We now present two algorithms for this inference problem. As with the Comparative Ball Algorithm, these are computationally simple to run, as we demonstrate in Section \ref{sec:simulations}, where we run them on large-size synthetic and real-world graphs. 

\subsubsection{The Threshold Ball Algorithm}

The Threshold Ball Algorithm is quite similar to the Comparative Ball Algorithm. Our goal is to return either INFECTION or RANDOM if the sickness is from an infection on $G$ or a random sickness respectively. It uses a threshold parameter, that represents the degree of clustering, where here we use the radius as a proxy for this level of clustering. This threshold may be calculated from the time $t$ if known, or estimated from the reporting sick nodes otherwise.

\begin{algorithm}
\caption{Threshold Ball Algorithm}
\textbf{Input:} Graph $G$; Set of reporting sick nodes $S_{{\rm rep}}$; Threshold $m$ \\
\textbf{Output:} INFECTION or RANDOM \\
\label{alg:ball}
\begin{algorithmic}
\STATE $k \gets {\rm RadiusBall}(G, S_{{\rm rep}})$
\IF{$k \leq m$}
\RETURN {INFECTION}
\ELSE
\RETURN {RANDOM}
\ENDIF
\end{algorithmic}
\end{algorithm}

\subsubsection{The Threshold Tree Algorithm}

The Threshold Tree Algorithm is similar, but rather than use ball-radius as a proxy for degree of clustering, it uses the weight of a minimum-weight spanning tree connecting all reporting infected nodes. We denote the weight of this tree on graph $G$ for set $S$ as {\rm SizeTree}(G, S). This algorithm also requires a threshold parameter. As before, the appropriate threshold may be calculated using the time $t$, or estimated from the set of reporting sick nodes.

\begin{algorithm}
\caption{Threshold Tree Algorithm}
\textbf{Input:} Graph $G$; Set of reporting sick nodes $S_{{\rm rep}}$; Threshold $m$ \\
\textbf{Output:} INFECTION or RANDOM \\
\label{alg:tree}
\begin{algorithmic}
\STATE $k \gets {\rm SizeTree}(G, S_{{\rm rep}})$
\IF{$k \leq m$}
\RETURN {INFECTION}
\ELSE
\RETURN {RANDOM}
\ENDIF
\end{algorithmic}
\end{algorithm}

\subsection{Summary of Results}
We analyze this inference problem and in particular the performance of our two algorithms, the Threshold Ball Algorithm and the Threshold Tree Algorithm, on three types of graphs. First, we consider an infection on a $d$-dimensional grid. In this case, both our algorithms are able to (asymptotically) eliminate Type I and Type II error, for up to a constant fraction of sick nodes, even when only a logarithmic fraction report sick. Orderwise, it is clear that this is the best any algorithm (regardless of computational complexity) can hope to achieve. Our empirical results verify this performance, and also show that the Ball Algorithm outperforms the Tree Algorithm on the grid. 

Next we consider tree graphs. Here we show that the Tree Algorithm can correctly discriminate between infections and random sickness for larger numbers of reporting sick nodes than the Ball Algorithm is able to handle. Finally, we analyze Erd\"{o}s-Renyi graphs under two different connectivity regimes: a low-connectivity with edge probability close to the regime when the giant component emerges; and a high connectivity regime the produces densely connected graphs. Again, we show that each algorithm can identify an infection with probabilities of error that decay to $0$ as the network size goes to infinity, for appropriate ranges of parameters. Not surprisingly, the more densely connected, the more difficult it becomes to obtain a good measure of `clustering.' Consequently, in these latter regimes, we find that one needs to intercept the sickness much earlier, i.e., with many fewer reporting sick nodes, in order to hope to accurately discriminate between the two potential sickness mechanisms. In the Erd\"{o}s-Renyi setting, we are unable to find direct analytic results to compare our two algorithms. However, in Section \ref{sec:simulations} we evaluate them empirically and find that the Ball Algorithm tends to perform better, despite its relative algorithmic simplicity.

\subsection{Multidimensional Grids}
\label{ssec:grid}

Let $G^{(n)}$ be a $n$-node $d$-dimensional grid network, with side length $n^{1/d}$. As before, to avoid edge effects, we let the opposite edges of the grid connect, so that the graph forms a torus, thereby eliminating any dependence of our results on the initial source of an infection. In this section, we show that both the Threshold Ball Algorithm and the Threshold Tree Algorithm can successfully distinguish an epidemic from a random illness, even when many nodes are infected, yet very few report the infection.

We consider first the Threshold Ball Algorithm. The key result here is the Shape Theorem given in Lemma \ref{lem:gridbound}, which, recall, essentially says that with high probability, the {\em shape} of the set of infected nodes closely resembles a ball. The key quantity, then, is the radius of this ball, i.e., the threshold the algorithm chooses in order to decide if the underlying cause of the illness is a spreading epidemic, or a random illness. 

Like before, we denote the set of reporting nodes $S_{{\rm rep}}{(n)}$. We first assume that in addition to the reporting likelihood, $q$, we know the time $t^{(n)}$ that has elapsed since the first infection (or, equivalently, the expected size of the infection). The threshold the algorithm uses is then a simple (linear) function of $t^{(n)}$. We then give an adaptive algorithm, that estimates $t^{(n)}$ and hence the optimal threshold to use, from the number of infected nodes reporting, and the reporting likelihood. We omit the superscript $n$ when it is clear from context.

The next result says that as long as the number of reporting sick nodes is at least $\log n$, then even if a constant fraction of nodes are infected, the Threshold Ball Algorithm can successfully distinguish the cause of the illness, provided that the time $t$ is known. We note that this requirement on the number of reporting sick nodes is essentially tight, i.e., the result cannot be improved orderwise. We also note that this requirement on the number of reporting nodes, along with the time $t$, implicitly constrains the underlying parameters of the problem setup, namely $q$. We also prove the algorithm succeeds under similar (but slightly more restrictive) conditions when $t$ is not known. We use $\mu$ to denote the expected rate that an infection travels along an axis on the grid. As remarked above, this rate $\mu$ is only a function of the dimension of the graph, since we assume the spreading rate to be normalized. We have the following.

\begin{thmsect}\label{thm:gridBall}
Consider the Threshold Ball Algorithm (Algorithm \ref{alg:ball}). Suppose that the expected number of reporting nodes scales at least as $\log n$.
\begin{enumerate}[(a)]
\item
Suppose $t$ is known. Set the threshold $m = 1.1 d \mu t$. Then if the expected number of infected nodes is less than $n/(4d)^d$,
$$
P({\rm error}) \rightarrow 0.
$$
\item
Next, suppose time $t$ is unknown. Let $X_{{\rm rep}}$ be the number of nodes reporting an infection, ${\rm card}(S_{{\rm rep}})$. Use threshold $m = 1.1 d (X_{{\rm rep}} \log \log n/q)^{1/d}$.  Then provided that the expected number of infected nodes is less than $n/((4d)^d \log \log n)$,
$$
P({\rm error}) \rightarrow 0.
$$
\end{enumerate}
\end{thmsect}

%
%

In other words, an infection can be identified in both cases with probability approaching $1$ as $n$ tends to infinity. Note that the guarantee is identical, up to the $\log \log n$ factor in the denominator; this is the price we pay for not explicitly knowing the initial time of the infection. 

\begin{proof}[Proof of Theorem \ref{thm:gridBall}(a)]

This proof follows along similar lines as those in Section \ref{ssec:grid}. First consider the Type II error probability, the probability a spreading infection is labeled a random sickness. This follows from the intuitive fact that an epidemic cannot spread at a rate that is a constant factor faster than $\mu$, its expected rate of spread. Indeed, from Proposition \ref{prop:gridInfectUB},
$$
{\rm RadiusBall}(G,S_{{\rm rep}}) < 1.1 d \mu t,
$$
with probability tending to $1$ as $n \rightarrow \infty$. This is equivalent to the Type II error probability tending to $0$.

Now consider the Type I error probability, namely that a random sickness is mistaken for an infection. From Proposition \ref{prop:gridRandLB}, since the number of reporting sick nodes, $S_{{\rm rep}}$, satisfy $S_{{\rm rep}} > \log n$, the smallest ball that contains these random nodes satisfies, with high probability,
$$
{\rm RadiusBall}(G,S_{{\rm rep}}) > n^{1/d}/4.
$$
Moreover, from the shape theorem of Lemma \ref{lem:gridbound}, we know that if the reporting sick nodes were in fact due to an epidemic, then nearly all the nodes within the smallest ball containing the reporting sick nodes, would in fact be sick. More precisely, Lemma \ref{lem:gridbound} says that given any radius a constant factor less than $n^{1/d}/4$, with high probability, all nodes inside that ball are infected. Thus, all nodes in a ball of radius $0.9 n^{1/d}/4$ would be infected, if the true infection mechanism were an epidemic. But this means that the total number of nodes actually infected is at least the number of nodes in this ball. By assumption, the expected number of infected nodes does not exceed $n (1/4d)^d$. Comparing these, we reach a contradiction.
Hence, the Type I error probability also tends to $0$.
\hfill \hspace{0.1cm} \end{proof}

We now use the previous result to prove that the adaptive threshold, where we use the number of reporting nodes to estimate $t$, also works. First we state a simple lemma to characterize the number of sick nodes.

\begin{lemma}
If at least $X$ nodes are sick, then the number of reporting nodes is at least $(1-\delta) q X$ with probability at least $1-\exp(-(1-\delta)^2 q X/2)$.
\end{lemma}
\begin{proof}
This is a well known Chernoff bound.
\end{proof}

Theorem \ref{thm:gridBall}(b) follows from this in a simple manner.

\begin{proof}[Proof of Theorem \ref{thm:gridBall}(b)]

Let $X_{{\rm rep}}$ be the number of reporting sick nodes, and let $\bar{X} = X_{{\rm rep}}/q$ (that is, $\bar{X}$ is basically the expected number of sick nodes based on the number reporting). From the previous lemma, we have
$$
P(\bar{X} \log \log n < {\rm card}(S)) \rightarrow 0.
$$
Let $\mu$ be the asymptotic rate at which an infection travels, as before. Let $\epsilon > 0$. From the proof of Theorem \ref{thm:gridBall}(a), at time $t$, we know for $\delta>0$
$$
P({\rm card}(S) < 2 (1-\epsilon) (\mu t)^d) \rightarrow 0.
$$
Hence $t < \frac{(\bar{X} \log \log n)^{1/d}}{\mu (2 (1-\epsilon))^{1/d}}$ with high probability. Naturally, increasing $t$ only increases the infection size, so it is only necessary to consider the maximum likely $t$. In particular, if the threshold $m = \mu t_{\max} = \frac{\mu (\bar{X} \log \log n)^{1/d}}{\mu (2 (1-\epsilon))^{1/d}}$, then from Theorem \ref{thm:gridBall}(a), the adaptive thresholding algorithm has Type I error probability approaching $0$. In addition, if $\bar{X}$ is $\omega(\log{n})$, the Type II error probability decays to $0$ as well, from the same theorem.
\hfill \hspace{0.1cm} \end{proof}

\subsection{Trees}
\label{ssec:tree}

We consider the problem on tree graphs. Unlike graphs (and more generally, geometric graphs), trees have exponential spreading rates, and hence manifest fundamentally different behavior. Indeed, while simple, tree graphs convey the key conceptual point of this section: the difficulty of distinguishing an epidemic from a random sickness on graphs where the infection spreads quickly. In addition, while the results do not immediately carry over, the behavior on a tree provides an intuition for the behavior of an infection on an Erd\"{o}s-Renyi graph, which we cover in the next section. 

Thus, let $G^{(n)}$ be a balanced tree with $n$ nodes, constant branching ratio $c \geq 2$, and a single root node. In the case of an infection, instead of choosing a node at random to be the original source of the infection, we always choose the root of the tree. This is the most interesting case, since otherwise a constant fraction of the nodes are very far from the infection source and bottlenecked by the root node. Also, this precisely models the scenario for locally tree-like graphs, such as Erd\"{o}s-Renyi graphs. We again omit the indexing on $n$ when it is clear by context.


First we examine the performance of the Threshold Ball Algorithm on this graph. Again recall the meaning of $t$: it is the time at which the sicknesses are reported, and also a proxy for the expected number of infected nodes.
\begin{thmsect}\label{thm:treeBall}
Suppose the Threshold Ball Algorithm (Algorithm \ref{alg:ball}) is used. Additionally, suppose $t$ is sufficiently large that the expected number of reporting nodes is at least $\log n$.
\begin{enumerate}[(a)]
\item
In the case $t$ is known, there exist constants $b$, $\beta$ such that if the expected number of infected nodes is less than $n^{\beta}$, then the tree algorithm with threshold $m = 1.1 b t$ succeeds:
$$
P({\rm error}) \rightarrow 0.
$$
\item
On the other hand, suppose $t$ is not known. Define $X_{{\rm rep}}$ as ${\rm card}(S_{{\rm rep}})$. Then there exists constants $b_2$ and $\beta$, with the threshold set $m = 1.1 b_2 \log (X_{{\rm rep}} (\log \log n)^2/q)$, where if the expected number of infected nodes is less than $n^{\beta}$,
$$
P({\rm error}) \rightarrow 0.
$$
\end{enumerate}
The constant $\beta$ is identical is both parts (a) and (b).
\end{thmsect}

%
%

\begin{proof}[Proof of Theorem \ref{thm:treeBall}(a)]
To prove this theorem, we prove the following more general statement:

For some constant $\beta < 1$, if $q E[{\rm card}(S)] = \omega(1)$ and $E[{\rm card}(S)] < n^{\beta}$, then the Type I error probability tends to $0$. Next, there exists a constant $b$ such that if $b_0 > b$ and the threshold $m > b_0 t$ for all $n$, then the Type II error probability converges to $0$ asymptotically, as the tree size scales.

The Type II error bound follows from results in first passage percolation \cite{itai94}. In particular, one can compute the fastest-sustainable transit rate. This quantity is basically the time from the root to the leaves, normalized for depth, as the size of the tree scales. Formally (again, see \cite{itai94} for details), let us consider a limiting process of trees whose size grows to infinity, with $\Gamma_n$ denoting the balanced tree on $n$ nodes, and $\delta(\Gamma_n)$ denoting the set of paths from the root to the leaves, and for a node $v \in p$ for some path $p \in \delta(\Gamma_n)$, let $T_v$ denote the time it takes the infection to reach node $v$. Then the {\em fastest-sustainable transit rate} is defined as:
$$
\lim_n \inf_{p \in \delta(\Gamma_n)} \limsup_{v \in p} \frac{T_v}{{\rm depth}(v)}.
$$ 
Basic results \cite{itai94} show that this quantity exists and is finite, and thus shows that the rate at which an infection travels, defined as the maximum distance of the infection from the root over time, converges to a constant $b$ that depends on the branching ratio. The probability that an infection travels at a faster rate converges to $0$ in the size of the tree. This establishes the Type II result.

The Type I error result follows simply as well. Given the branching ratio, $c$, there are $\frac{c^{m+1}-1}{c-1}$ nodes within a distance $m$ from the root. Again letting $S_{{\rm rep}}$ denote the number of reporting sick nodes, the probability of a Type I error is controlled by $(\frac{c^m}{n})^{S_{{\rm rep}}}$ -- the probability that the randomly sick nodes are closer than the threshold $m$ to the root. Then if $c^m$ is $o(n)$, it is sufficient that the probability that $S_{{\rm rep}} = 0$ goes to $0$. This occurs if the expected number of reporting sick nodes is $\omega(1)$. That is, we need $q E[{\rm card}(S)] = q e^{(c-1)t} = \omega(1)$, calculating $E[{\rm card}(S)]$ with a simple differential equation. Alternatively, if $c^m = \alpha n$ for some constant $\alpha < 1$, then we require $S_{{\rm rep}}$ to increase with $n$ with probability $1$. The same condition as before is sufficient for this to be true. This completes the Type I result.

Using both these results, there is a choice of $m$ such that both error types become rare as long as $c^{b_0 t} < \alpha n$, so $c^t < (\alpha n)^{1/b_0}$. The theorem follows using a particular threshold.
\hfill \hspace{0.1cm} \end{proof}

\begin{proof}[Proof of Theorem \ref{thm:treeBall}(b)]
First, note that $E[{\rm card}(S)]$ scales as $e^{(c-1)t}$. In fact, for any fixed $\epsilon > 0$, ${\rm card}(S) > e^{(c-1)t/(1+\epsilon)}$ with probability approaching $1$ (for example, see \cite{grey74}). Now we can proceed as in the proof of Theorem \ref{thm:gridBall}(b).

As before, let $X_{{\rm rep}}$ be the number of reporting sick nodes, and $\bar{X} = X_{{\rm rep}}/q$. Then we conclude $t_{\max} = (1+\epsilon)/(c-1) \log (\bar{X} (\log \log n)^2)$. Hence, by setting $b_2 = (1+\epsilon) b/(c-1)$, we see the Type II error probability converges to $0$ by Theorem \ref{thm:treeBall}(a). Using the same theorem, we see the Type I error also goes to $0$.
\end{proof}

Thus, the Threshold Ball Algorithm succeeds until the farthest infected node reaches the edge of the graph. At this point, the ball radius can increase no further, thus there is no hope of distinguishing an infection from a random sickness. Since this farthest point travels at a faster rate than the bulk of the infection, the Ball Algorithm can only work up to some time $\log_c{n}/b$. The Threshold Tree Algorithm, however, is better suited for this setting. We consider this next, and show that the Tree Algorithm can still correctly identify an infection with high probability nearly to the point where $\Theta(n)$ nodes are sick. This includes infection times close to $\log_c{n}$, the time it takes for every node to be infected. From this, we see that the Tree Algorithm works for a wider range of times compared to the Ball Algorithm. This is also demonstrated by simulations in Section \ref{sec:simulations}.

We note that the threshold in the results below on the Tree Algorithm, depend on $E[{\rm card}(S)]$ instead of depending explicitly on $t$, but as discussed previously, these are essentially equivalent, and we switch between the two merely to simplify notation and the exposition.

\begin{thmsect}\label{thm:treeTree}
Consider when the Threshold Tree Algorithm (Algorithm \ref{alg:tree}) is applied to this problem. Suppose $q > \log \log n/ \log n$, and $t$ is sufficiently large that the expected number of reporting nodes is at least $\log n$.
\begin{enumerate}[(a)]
\item
Consider when $t$ is known. Then for any constant $\alpha < 1$, if the expected number of infected nodes scales as less than $n^{\alpha}$, with threshold $m = E[{\rm card}(S)] \log \log n$,
$$
P({\rm error}) \rightarrow 0.
$$
\item
Suppose $t$ is not known. Set $X_{{\rm rep}}$ = ${\rm card}(S_{{\rm rep}})$, the number of nodes reporting an infection. Use threshold $m = X_{{\rm rep}}/q (\log \log n)^3$. Then if for any constant $\alpha < 1$, the expected number of infected nodes is less than $n^{\alpha}$,
$$
P({\rm error}) \rightarrow 0.
$$
\end{enumerate}
\end{thmsect}

%

\begin{proof}[Proof of Theorem \ref{thm:treeTree}(a)]
We prove the following generalization of the theorem:
The Type I error probability converges to $0$ for any choice of the threshold $m = o(q E[{\rm card}(S)] \log{n})$ with $q E[{\rm card}(S)] = O(n^{\alpha})$ for some $\alpha<1$. In addition, the Type II error probability converges to $0$ if $m = \omega(E[{\rm card}(S)])$.

First we prove the Type II error result (mistaking an infection for a random sickness).
Since the Steiner tree containing the reporting nodes can be no larger than the infection itself, the Type II error converges to $0$ as long as we use a threshold $m = \omega(E[{\rm card}(S)])$ from Markov's inequality.
Next, we evaluate the Type I error probability (mistaking a random sickness for an infection). This requires estimating the size of the Steiner tree containing the reporting sick nodes. By assumption, the number of reporting sick nodes increases with $n$, the probability that there are sick nodes on at least two subtrees of the root node goes to $1$, hence the root of the tree is in the Steiner tree connecting the randomly sick nodes with high probability. Given this, we see that a node is in the Steiner tree if and only if it is infected or a node below it in the tree is infected. By assumption, $E[{\rm card}(S_{{\rm rep}})] > \log n$. Let $X_{{\rm rep}} = {\rm card}(S_{{\rm rep}})$, and hence $X_{{\rm rep}}$ is $\omega(1)$. Choose the first level in the tree that has at least $X_{{\rm rep}}/c$ nodes. Then there are between $X_{{\rm rep}}/c$ and $X_{{\rm rep}}$ subtrees below that level. It is straightforward to show that each sick node in the tree has at least a $1/2$ probability of being a leaf node since $c \geq 2$. Since at least $X_{{\rm rep}}$ nodes are sick, at least $X_{{\rm rep}}/4$ of the leaf nodes are sick and distributed independently among the at most $X_{{\rm rep}}$ subtrees. Therefore, the total number of subtrees with sick nodes at the bottom is at least $X_{{\rm rep}}/(8c)$. In addition, each leaf node in a separate subtree requires a path at least up to the aforementioned level in the Steiner tree.  This gives us the following high probability bound on the Steiner tree size.
\begin{align}
{\rm SizeTree}(S_{{\rm rep}}) & > \frac{X_{{\rm rep}}}{8c}(\log_c{n}-\log_c{X_{{\rm rep}}}) \nonumber\\
 & > X_{{\rm rep}} \frac{(1-\alpha)\log_c{n}}{8c}. \nonumber
\end{align}
For any $w = o(E[X_{{\rm rep}}])$, we know that $X_{{\rm rep}} > w$ with probability approaching $1$ since the number of sick nodes in a random sickness is highly concentrated.
Therefore, if $m = o(E[X_{{\rm rep}}]\log_c{n})$, which is equivalent to $m = o(q E[{\rm card}(S)] \log{n})$, the Type I error probability tends to $0$.
\hfill \hspace{0.1cm} \end{proof}

\begin{proof}[Proof of Theorem \ref{thm:treeTree}(b)]
Let $X_{{\rm rep}} = {\rm card}(S_{{\rm rep}})$. Let $\bar{X} = X_{{\rm rep}}/q$, roughly the expected number of total sick nodes. Then $\bar{X} \log \log n$ upper bounds ${\rm card}(S)$ with high probability as shown previously. In addition, like before, ${\rm card}(S) \log \log n > E[{\rm card}(S)]$ with probability approached $1$. Then from Theorem \ref{thm:treeTree}(a) with $m = \bar{X} (\log \log n)^3$, we see that both probability of errors decrease to $0$ asymptotically.
\end{proof}

\subsection{Erd\"{o}s-Renyi Graphs}
\label{ssec:gnp}
In this section, we consider Erd\"{o}s-Renyi graphs. A notable difference in the topology of Erd\"{o}s-Renyi graphs and grids is that the diameter of the former scales much more slowly (logarithmically) with graph size. That is, Erd\"{o}s-Renyi graphs are more highly connected, in the sense that no two nodes are too far apart. This makes distinguishing an infection from a random sickness more difficult on these graphs.

We consider two connectivity regimes: the regime where the giant component first emerges, and each node has a constant expected number of edges, and then a much more highly connected regime, where the graph demonstrates different local properties, and discrimination between random sickness and infection is harder still.

\subsubsection{Detection with Constant Average Degree}
\label{sssec:moderatedetection}
We first consider Erd\"{o}s-Renyi graphs with constant average degree. Define the graph $G^{(n)} = G(n, p)$ to be the graph with $n$ nodes, where for each pair of nodes, there is an edge between them with probability $p$. In the section above, we use $c$ to denote the branching ratio. We overload notation and use it again to measure the spread of the graph, but here as the expected degree: let $p = c/n$ with $c > 1$. In this regime, the graph is almost surely disconnected, but there is a giant component. Since this problem would be trivial on a disconnected graph, we limit both the infection and random sick nodes to the giant component. We show that unlike the case of trees, our algorithms are unable to distinguish infection from random sickness when nearly a constant fraction of nodes are infected. Instead, we consider infections that cover only $o(n)$ nodes. As is well-known (e.g., \cite{durrett07}) in this connectivity regime, the graph is locally tree-like, and hence tree-like in the infected region. This allows us to leverage some results from the previous section, although direct translation is not possible, particularly in the analysis of our second algorithm. We will drop the index on $n$ for clarity.

Again we note that in the next two theorems, the threshold depends on $t$ and $E[{\rm card}(S)]$, respectively. As discussed, these are essentially equivalent, and the choice amounts to ease of notation and exposition.


\begin{thmsect}\label{thm:gnpBall}
Suppose we use the Threshold Ball Algorithm (Algorithm \ref{alg:ball}). Consider the case when the expected number of reporting nodes is no less than $\log n$.
\begin{enumerate}[(a)]
\item
Suppose we have knowledge of $t$. There are constants $b$, $\beta$ where, using threshold $m = b t$ and with expected number of infected nodes less than $n^{\beta}$,
$$
P({\rm error}) \rightarrow 0.
$$
\item
Consider unknown $t$. We set $X_{{\rm rep}}$ to be the number of nodes reporting an infection, ${\rm card}(S_{{\rm rep}})$. Then there exists constants $b_2$ and $\beta$ such that for threshold $m = b_2 \log (X_{{\rm rep}}/q (\log \log n)^2)$ and if the expected number of infected nodes is less $n^{\beta}$,
$$
P({\rm error}) \rightarrow 0.
$$
\end{enumerate}
The constant $\beta$ is the same for both (a) and (b).
\end{thmsect}

%

\begin{proof}[Proof of Theorem \ref{thm:gnpBall}(a)]

Consider the Type II error probability. In this case, from Proposition \ref{prop:gnpInfectUB}, there is a constant $b$ such that, with probability converging to $1$,
$$
{\rm RadiusBall}(G,S_{{\rm rep}}) < b t = m. \nonumber
$$
Therefore, the Type II error probability tends to $0$.

Now we bound the Type I error probability. From Proposition \ref{prop:gnpRandLB}, with probability tending to $1$,
$$
{\rm RadiusBall}(G,S_{{\rm rep}}) > \frac{\log n}{3 \log c}. \nonumber
$$
Therefore, it is sufficient to show $m < \frac{\log n}{3 \log c}$. Since the infection size is $o(n)$, we use a branching process approximation to find that for some $\lambda$, $E[{\rm card}(S)] \rightarrow e^{\lambda t}$. Define $\beta = \lambda/(3 b \log c)$. Since $E[{\rm card}(S)] < n^{\beta}$ by hypothesis,
$$
\lambda t < \beta \log n. \nonumber
$$
With some computation, $m = b t < \log n/(3 \log c)$. Hence, the Type I error probability also decays to $0$.

\hfill \hspace{0.1cm} \end{proof}

\begin{proof}[Proof of Theorem \ref{thm:gnpBall}(b)]
As is shown above, $E[{\rm card}(S)]$ scales asymptotically as $e^{\lambda t}$ for some constant $\lambda$. In particular, for abitrary constant $\epsilon > 0$, $E[{\rm card}(S)] > e^{\lambda t/(1+\epsilon)}$ with probability approaching $1$. Then let $X_{{\rm rep}}$ be the number of reporting sick nodes and let $\bar{X} = X_{{\rm rep}}/q$, so $\bar{X} \log \log n > {\rm card}(S)$ with probability tending to $1$ as shown previously. From this, we conclude $t_{\max} = (1+\epsilon)/\lambda \log (X_{{\rm rep}}/q (\log \log n)^2)$. Then by Theorem \ref{thm:treeBall}(a), with $b_2 = (1+\epsilon) b/\lambda$ and $m = b_2 \log (X_{{\rm rep}}/q (\log \log n)^2)$, we see that the Type II error probability converges to $0$. From the same theorem, the Type I error goes to $0$ as well.
\end{proof}

The Tree Algorithm is more complex to analyze for this graph. The more delicate analysis comes from the challenge of bounding the size of the Steiner tree for the random sickness process, needed to control Type I error.

\begin{thmsect}\label{thm:gnpTree}
Suppose the Threshold Tree Algorithm (Algorithm \ref{alg:tree}) is applied to this problem. Assume that the expected number of reporting nodes is at least $\log n$ and $q$ is constant.
\begin{enumerate}[(a)]
\item
Consider the case where $t$ is known. Let the threshold $m = E[{\rm card}(S)] \log \log n$. For any $\alpha < 1/2$, if the expected number of infected nodes scales as less than $n^{\alpha}$,
$$
P({\rm error}) \rightarrow 0.
$$
\item
Suppose we have unknown $t$. Define $X_{{\rm rep}}$ as ${\rm card}(S_{{\rm rep}})$. In this case, set the threshold to be $m = (X_{{\rm rep}}/q) (\log \log n)^3$. Then like before, for any constant $\alpha < 1/2$, if the expected number of infected nodes is less than $n^{\alpha}$, 
$$
P({\rm error}) \rightarrow 0.
$$
\end{enumerate}
\end{thmsect}

%

\begin{proof}[Proof of Theorem \ref{thm:gnpTree}(a)]
We show the following more general statement:
The Type II error probability decays to $0$ if the threshold is chosen as $m = \omega(E[{\rm card}(S)])$ and $E[{\rm card}(S)] = o(n)$. The Type I error probability goes to $0$ when $m < k q E[{\rm card}(S)]$ for some constant $k = o(\log (n/(q E[{\rm card}(S)])^2))$ and $q E[{\rm card}(S)] = o(\sqrt{n})$.

First, if the sickness is from an infection, the smallest tree connecting the reporting sick nodes must have size no more than the actual number of sick nodes. Hence, to bound the Type II error, it is sufficient to bound the probability the number of infected nodes is over a certain size. This probability decreases to $0$ as long as $m$ is $\omega(E[{\rm card}(S)])$ when $E[{\rm card}(S)] = o(n)$. To see this, recall that in this regime, the graph looks locally tree-like. Consequently, we can bound the maximum number of infected nodes using bounds on the distance an infection can travel (e.g., see \cite{itai94}). Again, Markov's inequality provides the exact error bound in the theorem statement.

To control Type I error probability, that a random sickness is mistaken for an infection, we must lower bound the size of the Steiner tree of a random sickness. For $v \in S_{{\rm rep}}$, let $d_v$ denote the distance from that node to the nearest other sick node. First we show that $\sum_{v \in S_{{\rm rep}}} d_v \leq 2 {\rm SizeTree}(G, S_{{\rm rep}})$. 
Note that the bound is attained for some graphs, such as a star graph with the central node uninfected.

Consider the Steiner tree subgraph, and duplicate all edges on it. Since the degree of each node in the subgraph is even, there is a cycle that connects all these nodes. Naturally, the length of this cycle, which is twice the size of the Steiner tree, is larger than the length of the smallest cycle connecting all sick nodes. In addition, the length of this cycle is at least $\sum_{v \in S_{{\rm rep}}} d_v$, since the distance from one sick node to the next sick node in the cycle is clearly no smaller than the distance from that sick node to the closest sick node. This establishes that $\sum_{v \in S_{{\rm rep}}} d_v\leq 2 {\rm SizeTree}(G, S_{{\rm rep}})$.

Now we simply need to bound $d_v$. To do this, we need an understanding of the neighborhood sizes in a $G(n,p)$ graph. But as the size of the graph scales, this is also straightforward to do: recalling that the probability of an edge is $c/n$ and hence the expected degree of each node is (asymptotically) $c$, then for typical nodes and arbitrary constant $\epsilon > 0$, there are no more than $\left((1 + \epsilon) c\right)^d$ nodes within distance $d$ provided that $d = \omega(1)$, using a branching process approximation.

Let $X_{{\rm rep}}$ be the number of reporting sick nodes. Now assume $X_{{\rm rep}} = o(\sqrt{n})$. Let $\epsilon > 0$ and $l = \epsilon n /X_{{\rm rep}}^2$. Let $k = o(\log (n/X_{{\rm rep}}^2))$. Using the above distance distribution calculation, we find that each sick node $v$, there are less than $l$ nodes within distance $k$. As the sick nodes are randomly selected, the probability that none of these are within a distance $k$ from $v$ is bounded by $(1-X_{{\rm rep}}/n)^l \rightarrow e^{-\epsilon/X_{{\rm rep}}} \rightarrow  1-\epsilon/X_{{\rm rep}}$. Thus the distance to the closest sick node to $v$ is at least $k$, i.e., $d_v > k$, with high probability, and using a simple union bound, the same is true, simultaneously, for all sick nodes. Hence the Steiner tree joining the set of {\em reporting} sick nodes is of size at least ${\rm SizeTree}(G, S_{{\rm rep}}) \geq (1/2)\sum d_v = (1/2) k q E[{\rm card}(S)]$, with probability decaying to zero. Therefore, the Type I error probability tends to $0$ as long as the threshold satisfies $m < k q E[{\rm card}(S)]/2$, for $k = o(\log(n/(qE[{\rm card}(S)])^2))$.
Using this result, we find that the Tree Algorithm can succeed so long as $q \log ({n/(qE[T])^2}) = \omega(1)$. This is a complex condition, though the conditions given in the theorem are sufficient for it to be true.
\hfill \hspace{0.1cm} \end{proof}

\begin{proof}[Proof of Theorem \ref{thm:gnpTree}(b)]
As in previous sections, we let $X_{{\rm rep}}$ be the number of reporting sick nodes, and define $\bar{X} = X_{{\rm rep}}/q$. Then as in Theorem \ref{thm:gnpTree}(a), $\bar{X} \log \log n$ upper bounds ${\rm card}(S)$ and ${\rm card}(S) \log \log n > E[{\rm card}(S)]$ with probability approaching $1$. Then from Theorem \ref{thm:gnpTree}(a), we see that for the specified threshold, both probability of errors decrease to $0$ asymptotically.
\end{proof}

\subsubsection{Detection on Dense Graphs}
\label{sssec:latedetection}

Now we consider the case of an Erd\"{o}s-Renyi graph with a denser set of edges. Higher connectivity means the infection spreads faster, making it more difficult to distinguish between spreading mechanisms. The performance depends critically on the exact scaling regime. We consider the regime where there exists $d \in \mathbb{Z}$ and constants $\epsilon, h \in \mathbb{R}$ such that $\epsilon < n^{d-1} p^{d} < h$ holds for all $n$ as $n \rightarrow \infty$. This connectivity regime has been studied in various places -- see, for example, \cite{distancedistribution} for further discussion of this scaling regime and properties of these dense graphs. The next result bounds the size of the Steiner tree on a random collection of nodes, and is the key result for bounding the Type I error. 

\begin{lemma}
Suppose nodes become sick, independently of each other, with probability $n^{1/d}/n$, so that the expected number of reporting sick nodes is $q n^{1/d}$. Further suppose $G = G(n, p)$ whose parameters satisfy $\epsilon < \lim_{n \rightarrow \infty} n^{d-1} p^{d} < h$ for $d > 4$. Let $Z$ be the size of the minimum Steiner tree connecting the reporting sick nodes. Also, let $m < (d-3) q n^{1/d}/2$ be the threshold for the Steiner tree size in the Tree Algorithm. Then $Z$ satisfies the following probabilistic limit:
$ \lim_{n \rightarrow \infty} Pr(Z < m) = 0$.
\label{thm:RandSteinerTree}
\end{lemma}

\begin{proof}
%
Using precisely the same argument as above, we can lower-bound the size of the Steiner tree by $\sum d_v \leq 2 Z$, where the sum is over all reporting sick nodes, and as before, $d_v$ denotes the minimum distance from a reporting sick node $v$ to the nearest other reporting sick node. To lower bound the size of this sum, we rely on a result from \cite{distancedistribution} that shows that in this scaling regime, the asymptotic distribution of the distance between two random nodes is positive on only $d$ and $d+1$. That is, {\em almost all nodes} are either at distance $d$ or $d+1$ from any given node $v$, and thus the distance distribution concentrates sharply around $d$. To put this another way, let $F_d$ be the probability that a random node is at distance more than $d$ from $A$. Then for any $\hat{d} > 1$, if $n^{\hat{d}-1} p^{\hat{d}} < h$, we have
$$
\lim F_{\hat{d}} = \lim_{n \rightarrow \infty} \exp^{-n^{\hat{d}-1}p^{\hat{d}}}.
$$
Recall $\lim_{n \rightarrow \infty} n^{\hat{d}-1}p^{\hat{d}}$ is bounded between $\epsilon$ and $h$.

Now we condition on the number of sick nodes, ${\rm card}(S)$. Using the same definite as before, let $X_{{\rm rep}}$ be the random variable with $X_{{\rm rep}} = {\rm card}(S_{{\rm rep}})$. Note $E[{\rm card}(S)] = n^{1/d}$ and the expected number of reporting sick nodes $E[X_{{\rm rep}}] = q E[{\rm card}(S)]$. We can compute the probability that the closest sick node is at distance more than $\hat{d}$ from a sick node $v$ simply as $F_{\hat{d}}^{X_{{\rm rep}}} \rightarrow \exp^{-(X_{{\rm rep}}/n) (np)^{\hat{d}}}$. Using our scaling regime, we know that $(\epsilon n)^{1/d} < np < (h n)^{1/d}$. To simplify notation, let $h' = h^{1/d}$. We have
\begin{align}
F_{d-3}^{X_{{\rm rep}}} &\rightarrow 1-X_{{\rm rep}}/n (np)^{d-3} \nonumber\\
 &> 1 - X_{{\rm rep}}/h' n n^{(d-3)/d}. \nonumber
\end{align}

Using a simple union bound, we find that the probability that some reporting sick node is within distance $d-3$ of another reporting sick node is at most $X_{{\rm rep}}^2/h' n n^{(d-3)/d}$. Since $X_{{\rm rep}}$ is a binomial random variable (since we condition on ${\rm card}(S)$), it concentrates about its mean: for any $\epsilon' > 0$, $Pr((1-\epsilon')E[X_{{\rm rep}}] < X_{{\rm rep}} < (1+\epsilon')E[X_{{\rm rep}}]) \rightarrow 1$. When $X_{{\rm rep}}$ is within this range, we find that $\sum d_v > (d-3)(1-\epsilon')E[X_{{\rm rep}}]$ with probability at least $1-(1+\epsilon')^2 h' E[X_{{\rm rep}}]^2 n^{-3/d} > 1-C n^{-1/d}$ for some constant $C$. This converges to $1$ for large enough $n$. Thus, we have shown the desired result. \hfill \hspace{0.1cm} \end{proof}

Now the probability of error calculations and hence the proof of correctness for the Tree Algorithm follows directly from the above.

%
%

\begin{thmsect}
For graph $G$ as above, suppose the expected number of reporting sick nodes is $q n^{1/d}$ and $t$ is known. Then for the Threshold Tree Algorithm, the probability of a Type I error converges to $0$, as long as the threshold satisfies $m < (d-3) q n^{1/d}/2$. The probability of a Type II error upper bounded by $2/(d-3-\epsilon)$ as long as the threshold satisfies $m > (d-3 - \epsilon) q n^{1/d}/2$, for any value of $\epsilon >0$ such that $\epsilon + 3 < d$. This bound converges to $0$ as $d \to \infty$.
\end{thmsect}

\begin{proof}
Consider first the probability of a Type I error. This is the probability that a random sickness has a Steiner tree of size less than $m$. From Theorem \ref{thm:RandSteinerTree}, this probability converges to $0$ if $E[{\rm card}(S)] = O(n^{1/d})$.

Second, consider the probability of a Type II error. As we have argued before, the size of this tree is no more than the total number of infected nodes, so it is sufficient to find the probability there are more than $m$ infected nodes. The Type II error probability bound follows from using Markov's Inequality. \hfill \hspace{0.1cm} \end{proof}

\section{Simulations}
\label{sec:simulations}
The above sections give theoretical guarantees for the correctness of our algorithms, and thus characterize their ability to distinguish the cause of an illness -- be it detecting one graph versus another as the causative network, or the determination that a sickness is an epidemic or a random illness. In this section, we explore these questions empirically. We validate our theoretical analysis on graphs that are generated from the ensembles we address in our theorems (grids, random graphs, trees) and then also consider epidemics on real-world graphs, and demonstrate that on these topologies as well, our algorithms perform well.

\subsection{Graph Comparison}

We simulated the performance of the Comparative Ball
Algorithm to evaluate the performance empirically. We determined the
error rate over a range of $t$ for several pairs of graphs. We
evaluated the two different standard graph topologies considered
earlier, grids and Erd\"{o}s-Renyi graphs. 

We simulated the infections on various pairs of the graphs over a
range of times. In order to portray the results in a comparable way,
we plotted the error rate versus the average infection size instead of
time. This is necessary because different times result in very
different infection sizes for the different graphs. That is, the
infection is large even at low $t$ on an Erd\"{o}s-Renyi graph, and
vice versa for a grid graph. This would introduce a misleading effect
in the results. 

Each node in the graphs received
a random label to ensure independence. We use $n = 1,600$ for each
graph with $q=0.25$. For the Erd\"{o}s-Renyi graphs, we use
$p=2/1,600$. The probability of error was computed over $10,000$
trials. There are two possible types of errors in each simulation,
when the infection spreads on the first graph, and when it spreads on
the second. We label the error event `T:$G_1$; A:$G_2$' for the error
where the infection in fact travels on graph $G_1$ (True event), but
the algorithm incorrectly labels it as occurring on graph $G_2$
(Algorithm output).


\begin{figure}[t!]
\centering
\includegraphics[height=7.5cm]{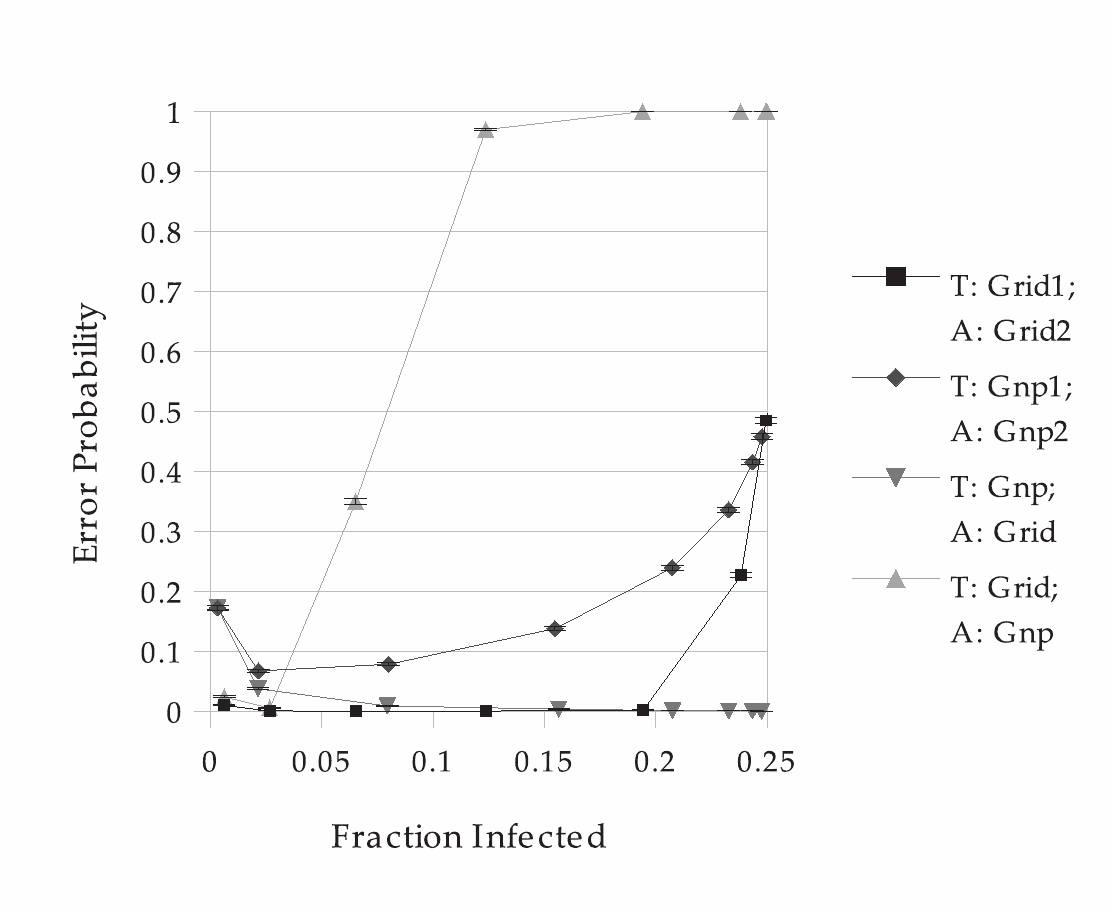}
\caption{{\footnotesize This figure shows the error
    probability for the algorithm on pairs of standard graphs. Various
    (conditional) error probabilities are illustrated -- `T:'
    corresponds to the true network, and `A:' corresponds to the
    algorithm output.}
\label{fig:T_Standard}}
\end{figure}

The results of these simulations are shown in Figure
\ref{fig:T_Standard}. Note that up to about $5\%$ of the network
reporting an infection, the error rates are low in all cases. The
error rates are consistently low for the `T:Grid1;A:Grid2' comparison
up to the point where the whole network is infected. When comparing a
grid and an Erd\"{o}s-Renyi graph, there is a bias to label it an
Erd\"{o}s-Renyi graph at higher times, causing the `T:Grid;A:G(n,p)'
error to be very high and conversely, the `T:G(n,p);A:Grid' error to
be very low. This suggests that by simply modifying the Comparative
Ball Algorithm to normalize with respect to a scaled graph diameter
(where the scaling parameter would be graph dependent), we could
balance these two error probabilities, and thus result in improved
performance. To illustrate, by choosing a diameter scaling value of $1.6$ for
the Grid graph, the plot in Figure~\ref{fig:scaled} indicates that
one could distinguish between G(n,p) and Grid graphs for a
significantly larger range. 
\begin{figure}[t!]
\centering
\includegraphics[height=7.5cm]{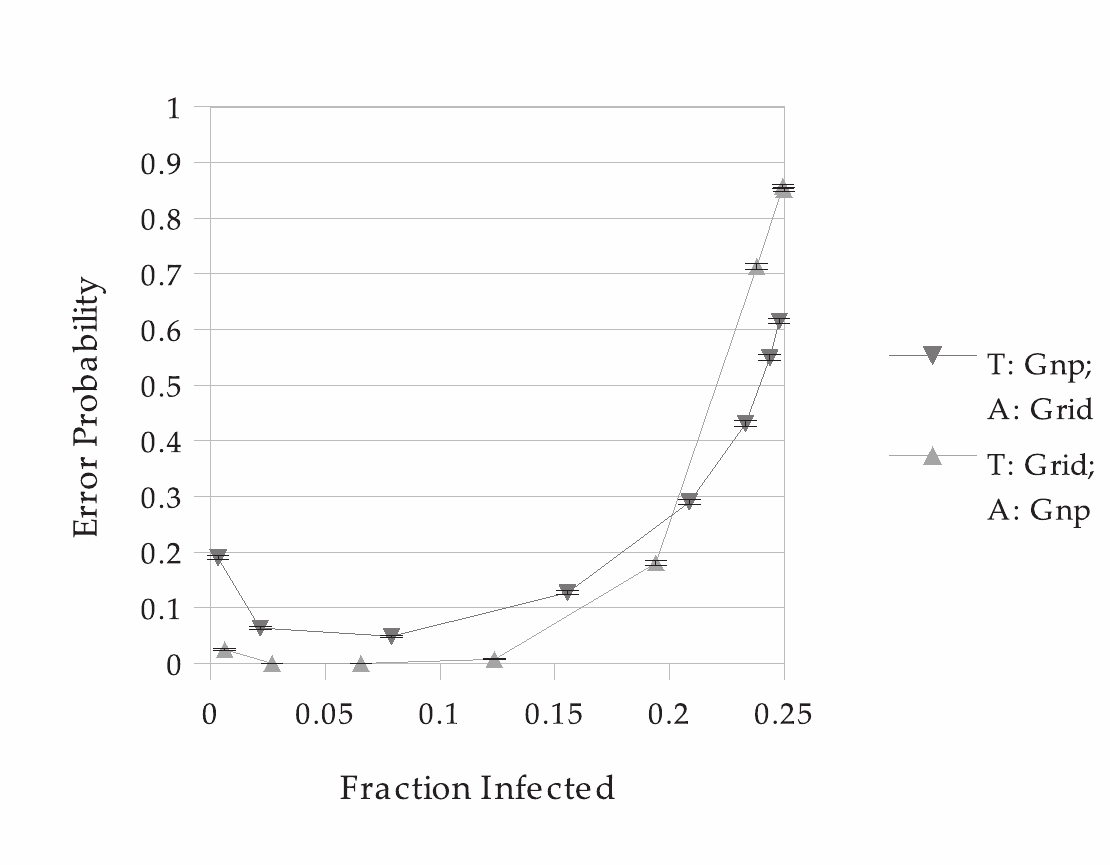}
\caption{{\footnotesize This figure shows the error
    probability for the G(n,p) vs. Grid graphs for the scaled diameter
    setting (diameter of G(n,p) graph is scaled by 1.6).}
\label{fig:scaled}}
\end{figure}
We plan to study a systematic approach for such scalings
as future work.

\subsection{Infection vs. Random Sickness}

In this section we provide simulation-based evidence of the theoretical results for the Threshold Ball Algorithm and Threshold Tree Algorithm. The simulations aim to demonstrate, in particular, two facts. First, the thresholds specified in Section \ref{sec:infectvsrandom} do actually work empirically, and as the graph size increases, the probability of both types of error decrease to zero. In addition, this provides insight into how quickly the probability of error decays. While our results include rate estimates given as part of the proof of correctness, we have not made an effort to optimize these in this work. Next, we seek to describe the relative performance of each algorithm, and show that it is as described above. Thus, we show that the Threshold Ball Algorithm outperforms the Threshold Tree Algorithm on a grid; the Threshold Tree Algorithm performs better than the Threshold Ball Algorithm on a balanced tree; and on an Erdos-Renyi graph, the performances are similar, with the Threshold Ball Algorithm performing slightly better. We accomplish this by determining the probability of error for a range infection sizes. The larger the fraction of infected nodes, the more difficult the problem becomes; hence we call an algorithm superior if it works for a larger fraction of infected nodes.

We note that to perform our simulations, it was necessary to use an approximate Steiner tree algorithm to perform the Threshold Tree Algorithm in a reasonable time frame. Naturally, since the exact problem is NP-hard, this would be required in any practical use of this algorithm at the moment. However, as a consequence, the empirical results may differ from the true theoretical result that would be obtained by employing an exact algorithm. Nevertheless, approximation algorithms typically have reasonable performance and we do not expect significant deviation from the correct results. The approximation algorithm we use is the Mehlhorn 2-approximation algorithm provided by the Goblin library \cite{mehlhornalg}. This algorithm is an efficient algorithm which produces a Steiner tree with no more than twice the optimal number of edges.

Each of the points in these results represents the average of $10,000$ runs. The average infection size, which is used to normalize the expected infection size in a random sickness, was determined by averaging the results of $10,000$ infections. For each simulation, we use a reporting probability $q = 0.25$, and other parameters ($n$, $t$ and $m$) as specified in each section below. Finally, the graphs are plotted with error bars at 95\% confidence.



\begin{figure}[h]
\centering
\includegraphics[height=5.2cm]{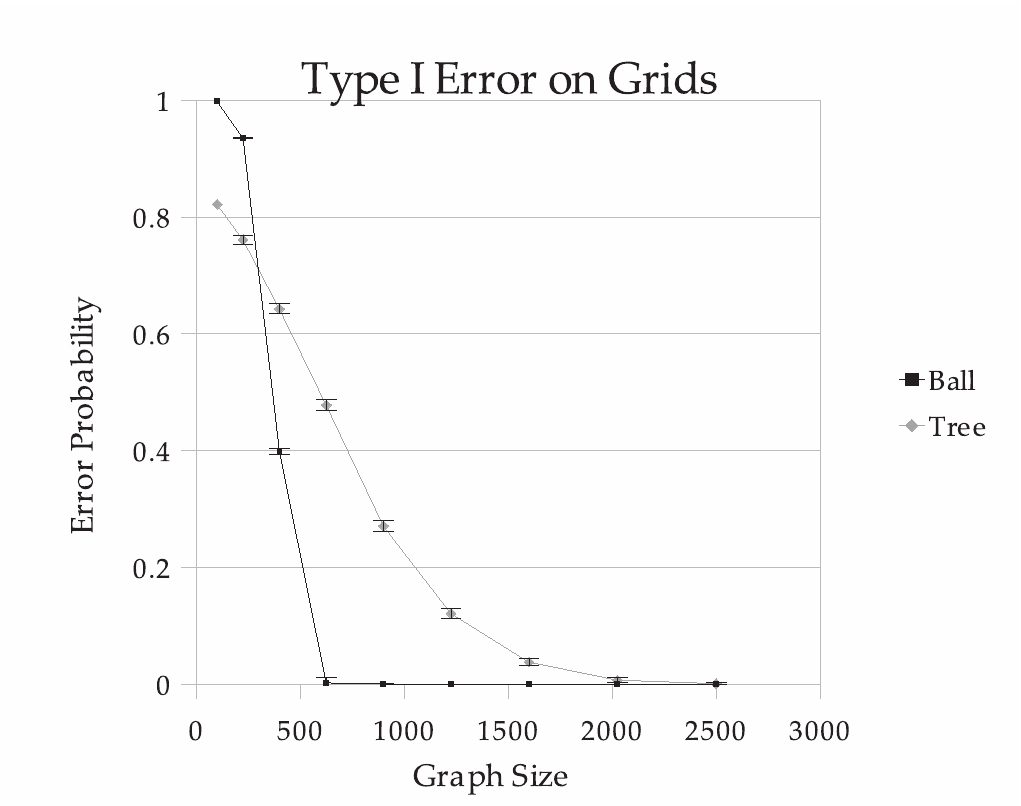} \includegraphics[height=5.2cm]{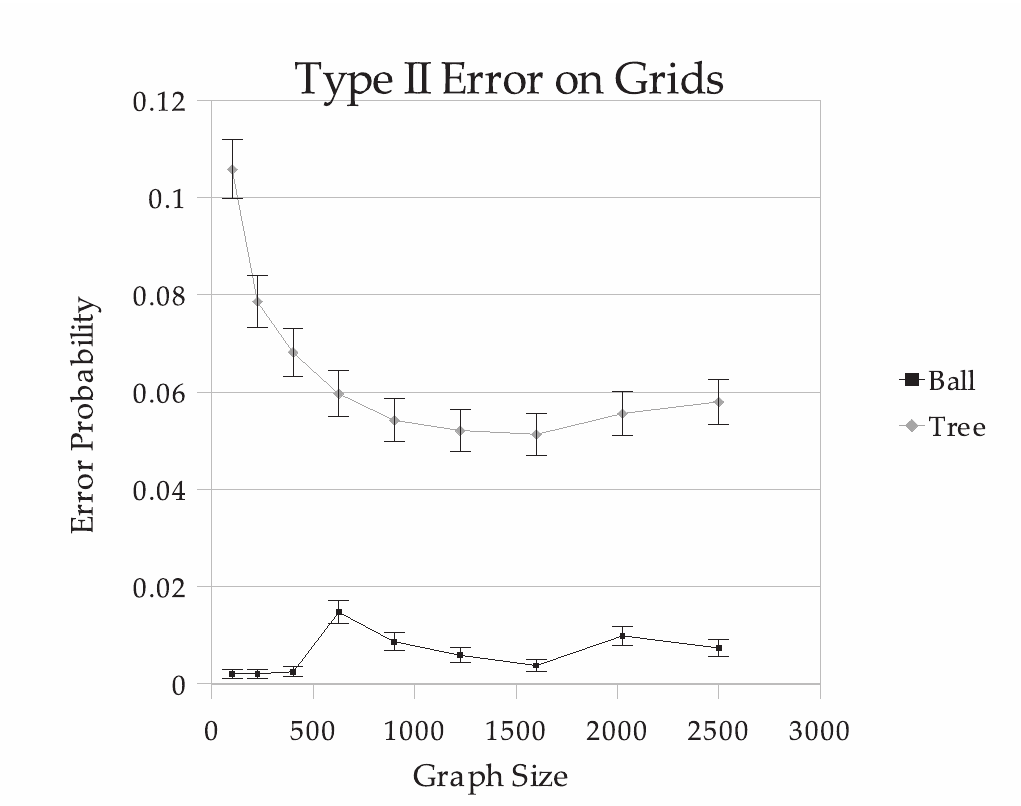} 
\caption[Grid graph Error.]{{\footnotesize Empirical Type I and Type II error probability vs graph size for grid graphs. The sample size is $10,000$ and infection size scales linearly with $n$.}
\label{fig:N_GridI}}
\end{figure}

\begin{figure}[h]
\centering
\includegraphics[height=5.2cm]{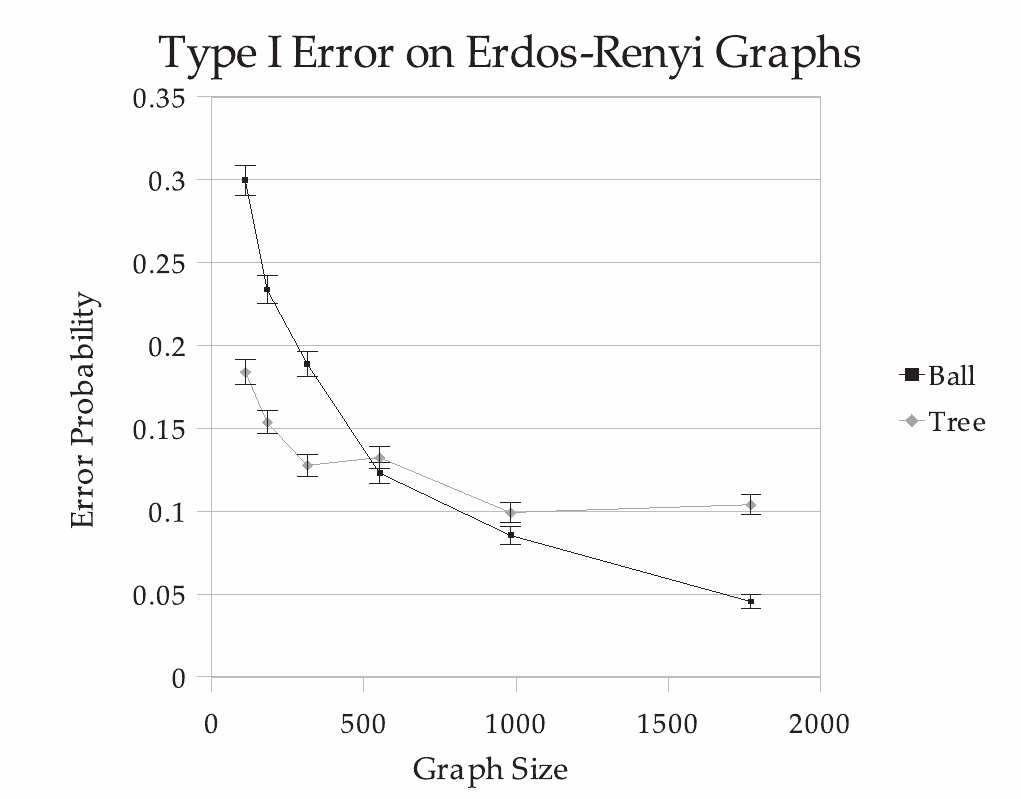} \includegraphics[height=5.2cm]{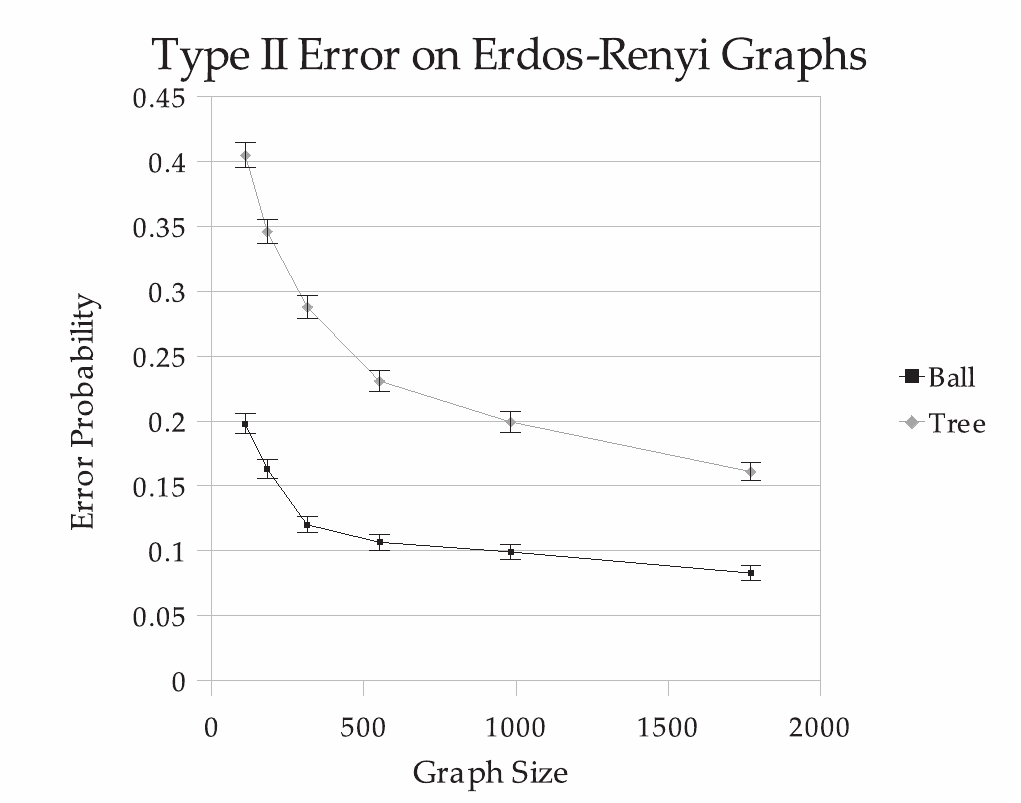} 
\caption[Erd\"os-Renyi Error.]{{\footnotesize Empirical Type I error probability vs graph size for graphs $G(n, 2/n)$. The sample size is $10,000$ and infection size scales orderwise as $\sqrt{n}$.}
\label{fig:N_GnpI}}
\end{figure}

%

%

\subsubsection{Error Rate Versus Graph Size}
Though our theoretical results have characterized the range for which each algorithm works, naturally we wish to see empirically the error probability for each algorithm and the rate at which the error decreases as graph size increases. 
Both Type I and Type II error probabilities were determined for each algorithm and graph topology. For this section, we have chosen time to keep the fraction of infected nodes at a consistent scaling. In particular, $t = 0.2 \sqrt{n}$ for the grid, and $t = 0.5 \log (0.5 n)$ with $p = 2/n$ for the Erd\"os-Renyi graph. The exact constants for these scalings were chosen empirically so that the probability of error was low and the Type I and Type II errors were as balanced as possible. The thresholds $m$ were also chosen with the same scaling, according to our theoretical results. To be exact, for the grid, the Threshold Ball Algorithm used threshold $m = 0.75 \sqrt{n}$ and the Threshold Tree Algorithm used threshold $m = 0.28 n$. For the Erd\"os-Renyi graphs, the Threshold Ball Algorithm used threshold $m = 0.69 \log (4.33 n)$ and the Threshold Tree Algorithm used threshold $m = 0.03 \sqrt{n \log n} \log n$. 

Figure \ref{fig:N_GridI} presents our results for grid graphs. The error probability of the Threshold Ball Algorithm on a grid is very low, while the tree algorithm performs relatively poorly. This is expected since the Threshold Ball Algorithm is closely aligned with the true shape of an infection on this graph. The Threshold Tree Algorithm has a much higher error probability which decays slowly with $n$, in particular the Type II error. 

Next, the results for Erd\"{o}s-Renyi graphs are in Figure \ref{fig:N_GnpI}. Here we see again that the Threshold Ball Algorithm performs better than the Threshold Tree Algorithm, at least for larger $n$, and that the error probability also seems to be decreasing faster for the Threshold Ball Algorithm as well. Though a tree more closely matches the infection shape on an Erd\"{o}s-Renyi graph, it is also easier for a random sickness to mimic a small tree, especially for small world graphs like Erd\"{o}s-Renyi graphs. This causes the Threshold Ball Algorithm to be ultimately superior. The Threshold Tree Algorithm is superior for larger infection sizes on bottle necked graphs (such as trees) where the random sickness can be easily distinguished, as we see in Section \ref{ssec:ErrorVsT}.


\begin{figure}[t!]
\centering
\includegraphics[height=5.2cm]{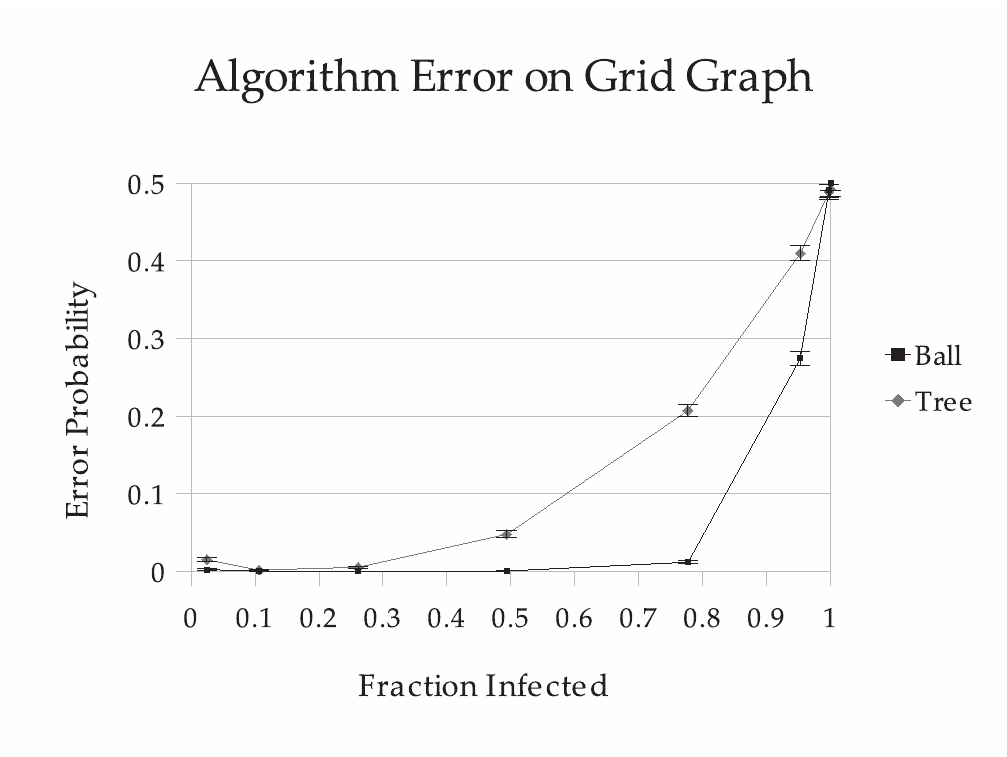} \includegraphics[height=5.2cm]{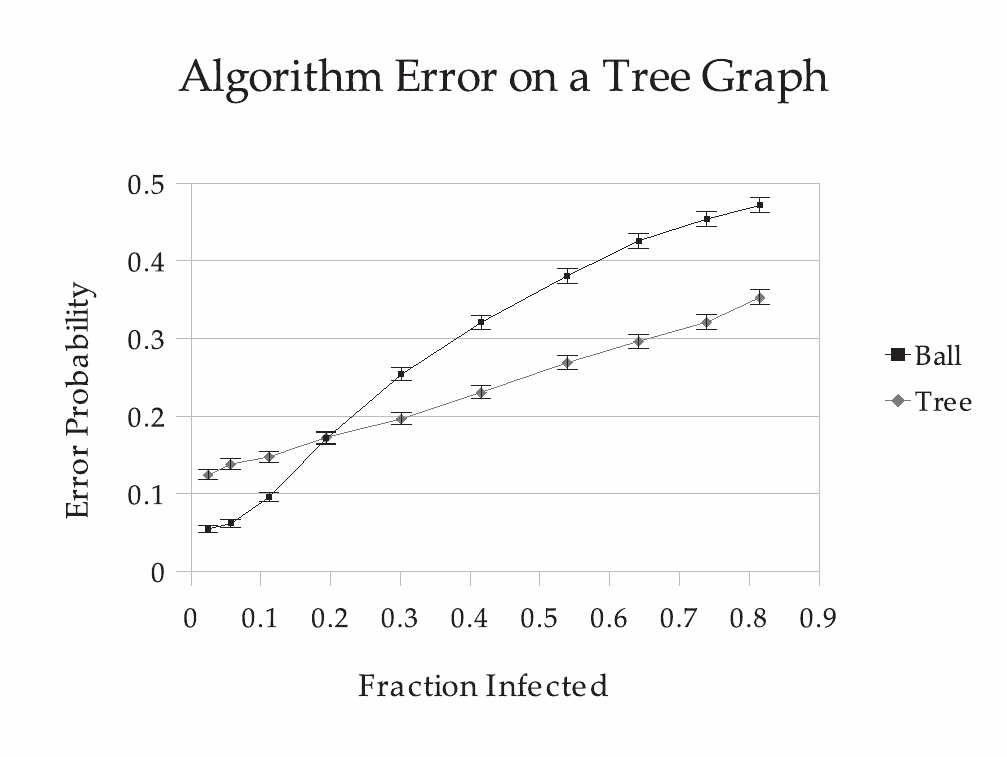} \includegraphics[height=5.2cm]{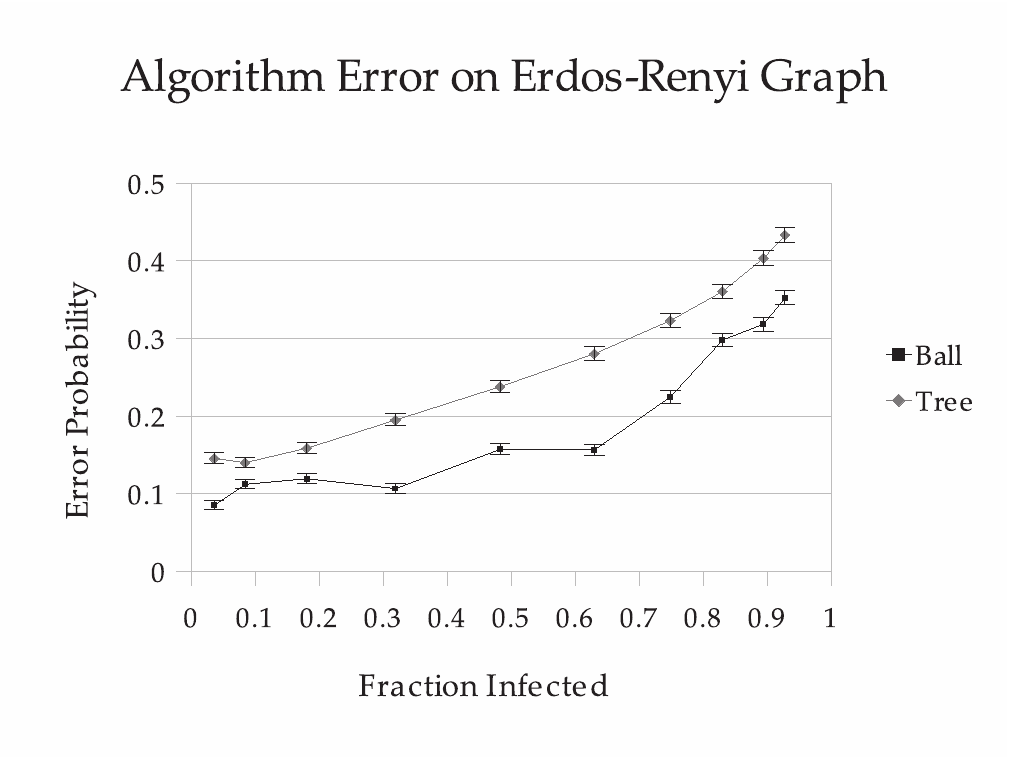} 
\caption[Algorithm Error.]{{\footnotesize This figure shows the overall error probability for each algorithm, for each of the three topologies we consider.}
\label{fig:T_Grid}}
\end{figure}

\subsubsection{Error Rate Versus Infection Size}
\label{ssec:ErrorVsT}
Next, we examine empirically how the infection duration affects the probability of error for each of our algorithms. As discussed above, we compare the two algorithms by the range of infection sizes for which they work, and accordingly, we call an algorithm superior if it maintains a lower probability of error for a larger infection size (fraction of total infected nodes). We use thresholds that minimize the empirical overall probability of error. That is, the sickness was chosen to be either an infection or simply random with equal probability, and the threshold with minimum probability of error from the simulations was chosen.

These results are presented in Figure \ref{fig:T_Grid} for grids, trees, and Erd\"{o}s-Renyi graphs. For each of the graph topologies, we used a graph size of $n = 1,600$. The error probability is plotted against the average infection size from the simulation. This choice better conveys how infection size affects the error rate, which is the chief question of interest.


These charts allow us to compare the performance of the algorithms. It is clear that the error probability of the Threshold Ball Algorithm is less than that of the Threshold Tree Algorithm on both the grid and Erd\"{o}s-Renyi graphs. On these graphs, the Threshold Ball Algorithm performs uniformly better across variations in fraction of nodes infected. However, the results on a tree are more complex. When the total infection is small, the Threshold Ball Algorithm has superior performance. However, as a larger fraction of the network becomes infected, the Threshold Tree Algorithm has better performance. We believe it is this right tail that is most significant. In the regime where many of the nodes are infected, the infection is likely to have reached some of the leaves by this time, thus explaining the superiority of the Threshold Tree Algorithm in this regime.
However, many practical applications of these algorithms would occur when the infection is still of limited size, in which case the Threshold Ball Algorithm would perform better. The best algorithm would depend on the circumstances.

\begin{figure}[t!]
\centering
\includegraphics[height=5.2cm]{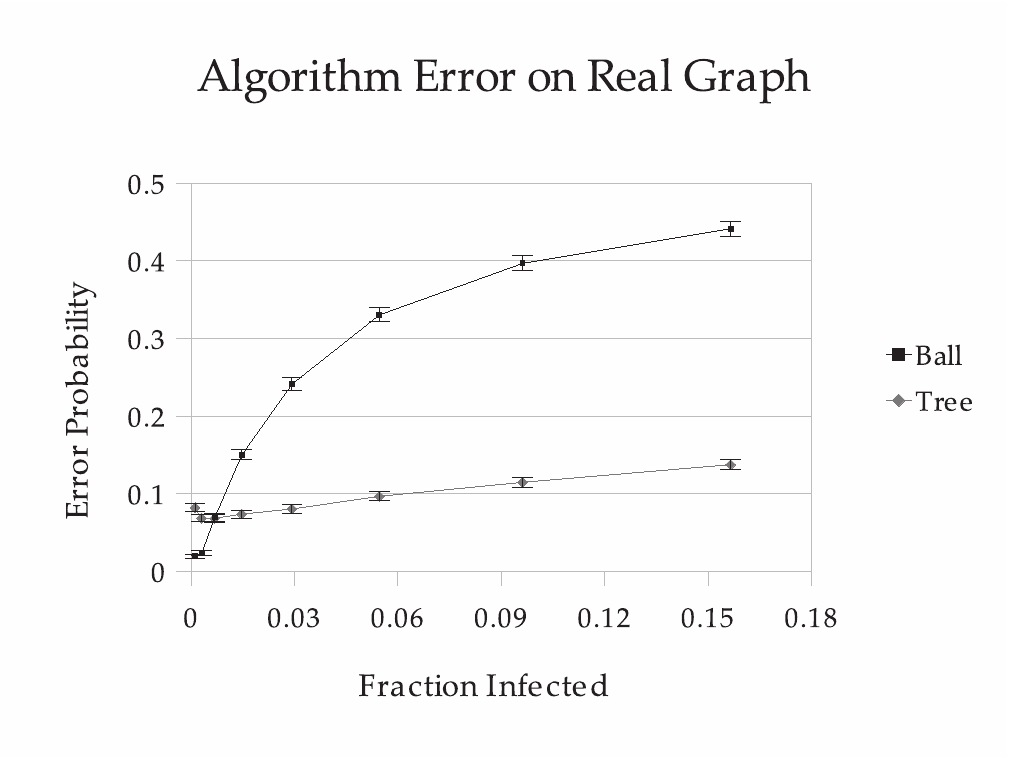}
\caption[Algorithm Error.]{{\footnotesize This figure shows the overall error probability for each algorithm on a real world graph.}
\label{fig:T_Real}}
\end{figure}

It is particularly interesting to ask how these results extend to real-world graphs, as opposed to random (or highly regular) graphs that we have constructed. To this end, we used the call-graph from an Asian telecom network. In this graph, each node is a cell customer, and there is an edge between two users if they contacted each other over this network during a certain range of time. Since the original graph was too large for practical simulation times, we cut out a partial subset. We chose a random node and all nodes with a distance $9$ and used the induced subgraph generated by these nodes. The resulting graph has size $n = 13,189$. The probability of error for a range infection sizes are presented in Figure \ref{fig:T_Real}. We see that the results are similar to those for a Tree graph, where the Threshold Ball Algorithm performs better on small infections, but it is out performed by the Threshold Tree Algorithm in larger infections. This is to be expected, as the intuition for the Threshold Ball Algorithm stems from the geometry of spatial grid-like networks. The call-graph here is very much tree-like (however, with very small diameter and high degree), and infections are unlikely to propagate to the same depth across various leaves. This results in poor Ball ``fits,'' especially as the infected fraction of nodes grows. This intuition is indeed borne out in the simulations.
%
%

\section{Conclusions}
When an infection/virus is seen spreading over a group of
people/machines, one may have multiple possibile spreading regimes for
the infection in mind, and want to know which the infection is most
likely travelling on. We considered this problem both in the case of two well structured graphs, and in the case of comparing an infection from a random sickness.
For two structured graphs, we have shown that this is possible to do with
high accuracy if the regimes are independent and satisfy two
properties: 1) An infection spreading according the regime should be
localized in the contact graph, and 2) A random set of nodes should be
spaced far apart on the graph. When these conditions are satisfied (in
the sense given in this paper), the correct spreading regime can be
detected accurately with high probability by determining on which
graph the infection appears to be more clustered. In addition, we have
shown two standard types of graphs, grids and Erd\"{o}s-Renyi graphs,
satisfy these properties.
In the case of comparing an infection and a random sickness, we developed two algorithms that solve the problem. We proved these algorithms do so with high probability for grids, tree, and Erd\"{o}s-Renyi graph for ranges of infection sizes dependent on the graph topology.
Our simulations here demonstrate the efficacy of our algorithms. 



\bibliographystyle{IEEEtran}
\bibliography{infectdetect}

\begin{thebibliography}{10}
\providecommand{\url}[1]{#1}
\csname url@samestyle\endcsname
\providecommand{\newblock}{\relax}
\providecommand{\bibinfo}[2]{#2}
\providecommand{\BIBentrySTDinterwordspacing}{\spaceskip=0pt\relax}
\providecommand{\BIBentryALTinterwordstretchfactor}{4}
\providecommand{\BIBentryALTinterwordspacing}{\spaceskip=\fontdimen2\font plus
\BIBentryALTinterwordstretchfactor\fontdimen3\font minus
  \fontdimen4\font\relax}
\providecommand{\BIBforeignlanguage}[2]{{%
\expandafter\ifx\csname l@#1\endcsname\relax
\typeout{** WARNING: IEEEtran.bst: No hyphenation pattern has been}%
\typeout{** loaded for the language `#1'. Using the pattern for}%
\typeout{** the default language instead.}%
\else
\language=\csname l@#1\endcsname
\fi
#2}}
\providecommand{\BIBdecl}{\relax}
\BIBdecl

\bibitem{sigmetrics2012}
C.~Milling, C.~Caramanis, S.~Mannor, and S.~Shakkottai, ``Network forensics:
  random infection vs spreading epidemic,'' \emph{SIGMETRICS Perform. Eval.
  Rev.}, vol.~40, no.~1, pp. 223--234, June 2012.

\bibitem{allerton2012}
------, ``On identifying the causative network of an epidemic,'' in \emph{In
  Proceedings of 50th Annual Allerton Conference on Communication, Control, and
  Computing}, October 2012.

\bibitem{aids-wiki}
\BIBentryALTinterwordspacing
Wikipedia, ``{HIV}/{AIDS} --- {W}ikipedia{,} the free encyclopedia,'' 2012,
  [Accessed 30-Sept-2012]. [Online]. Available:
  \url{http://en.wikipedia.org/wiki/HIV/AIDS}
\BIBentrySTDinterwordspacing

\bibitem{science06}
J.~Cohen, ``Making headway under hellacious circumstances,'' \emph{SCIENCE},
  vol. 313, pp. 470--473, July 2006.

\bibitem{massganesh05:epidemics}
A.~J. Ganesh, L.~Massouli{\'e}, and D.~F. Towsley, ``The effect of network
  topology on the spread of epidemics,'' in \emph{INFOCOM}, 2005, pp.
  1455--1466.

\bibitem{ball04}
F.~Ball and P.~Neal, ``Poisson approximation for epidemics with two levels of
  mixing,'' \emph{The Annals of Probability}, vol.~32, no.~1B, pp. 1168--1200,
  2004.

\bibitem{gopalan11}
A.~Gopalan, S.~Banerjee, A.~Das, and S.~Shakkottai, ``Random mobility and the
  spread of infection,'' in \emph{Proc. IEEE Infocom}, 2011.

\bibitem{demiris05a}
N.~Demiris and P.~D. O'Neill, ``Bayesian inference for epidemics with two
  levels of mixing,'' \emph{Scandinavian Journal of Stat.}, vol.~32, pp.
  265--280, 2005.

\bibitem{streftaris02}
G.~Streftaris and G.~J. Gibson, ``Statistical inference for stochatic epidemic
  models,'' in \emph{Proc. 17th International Workshop on Statistical
  Modeling}, 2002, pp. 609--616.

\bibitem{demiris05b}
N.~Demiris and P.~D. O'Neill, ``Bayesian inference for stochastic multitype
  epidemics in structured populations via random graphs,'' \emph{Journal of the
  Royal Stat. Society Series B}, vol.~67, no.~5, pp. 731--745, 2005.

\bibitem{infectionsource}
D.~Shah and T.~Zaman, ``Detecting sources of computer viruses in networks:
  Theory and experiment,'' \emph{SIGMETRICS Perform. Eval. Rev.}, vol.~86, pp.
  203--214, 2010.

\bibitem{shza11}
------, ``Rumors in a network: Who's the culprit?'' \emph{IEEE Transactions on
  Information Theory}, vol.~57, August 2011.

\bibitem{rwreandfpptrees}
R.~Lyons and R.~Pemantle, ``Random walk in a random environment and
  first-passage percolation on trees,'' \emph{The Annals of Probability},
  vol.~20, no.~1, pp. 125--136, 1992.

\bibitem{gridinfectionlimit}
H.~Kesten, ``On the speed of convergence in first-passage percolation,''
  \emph{The Annals of Applied Probability}, vol.~3, no.~2, pp. 296--338, Nov
  1993.

\bibitem{itai94}
I.~Benjamini and Y.~Peres, ``Tree-indexed random walks on groups and first
  passage percolation,'' \emph{Probability Theory and Related Fields}, vol.~98,
  pp. 91--112, 1994.

\bibitem{durrett07}
R.~Durrett, \emph{Random Graph Dynamics}.\hskip 1em plus 0.5em minus
  0.4em\relax Cambridge University Press, 2007.

\bibitem{neighborhoodbound}
F.~Chung and L.~Lu, ``The diameter of sparse random graphs,'' \emph{Adv. in
  Appl. Math}, vol.~26, pp. 257--279, 2001.

\bibitem{grey74}
D.~R. Grey, ``Asymptotic behaviour of continuous time, continuous state-space
  branching processes,'' \emph{Journal of Applied Probability}, vol.~11, no.~4,
  pp. 669--677, December 1974.

\bibitem{distancedistribution}
V.~D. Blondel, J.-L. Guillaume, J.~M. Hendrickx, and R.~M. Jungers, ``Distance
  distribution in random graphs and application to network exploration,''
  \emph{Physical Review}, vol.~76, no. 066101, 2007.

\bibitem{mehlhornalg}
K.~Mehlhorn, ``A faster approximation algorithm for the steiner problem in
  graphs,'' \emph{Information Processing Letters}, vol.~27, pp. 125--128, 1988.

\end{thebibliography}

\end{document}